\documentclass[a4paper,11pt]{article}

\usepackage[utf8]{inputenc}
\usepackage[T1]{fontenc}
\usepackage{fix-cm} 
\usepackage{a4wide}
\usepackage{amssymb,amsmath,xspace}
\usepackage{thm-restate} 
\usepackage{fixmath,bbold} 
\usepackage[scr]{rsfso} 
\usepackage[dvipsnames,table]{xcolor} 

\usepackage{graphicx}
\def\pathfig{figures}
\def\pathfig{.} 
\graphicspath{{\pathfig/}}

\definecolor{MyBlue}{RGB}{15,60,210}
\usepackage[pdftex]{hyperref}
\hypersetup{
  pdfauthor={Cyril Gavoille, Claire Hilaire},
  colorlinks=true,
  linkcolor=MyBlue,
  citecolor=MyBlue,
  urlcolor=MyBlue
}
\usepackage{enumitem} 
\setlist[itemize]{label=\textbullet} 
\usepackage{tikz} 
\usetikzlibrary{math}
\usetikzlibrary{calc}
\usepackage{pgfplots}
\usepgfplotslibrary{colormaps}
\pgfplotsset{
  compat=newest,
  colormap={mycolormap}{color=(lightgray) color=(white) color=(lightgray)}
}

\usepackage{caption}
\usepackage{subcaption}
\usepackage{comment}
\usepackage[capitalize]{cleveref}

\newtheorem{claim}{Claim}
\newtheorem{lemma}{Lemma}

\newtheorem{theorem}{Theorem}
\newtheorem{property}{Property}
\newtheorem{proposition}{Proposition}

\def\qed{{}\hfill$\Box$\\}
\newenvironment{proof}{\noindent\textbf{Proof.}}{\qed}
\newenvironment{proofof}{\noindent\textbf{Proof of~}}{\qed}

\newcommand{\card}[1]{\left|{#1}\right|}

\newcommand{\set}[1]{\left\{{#1}\right\}}
\newcommand{\range}[2]{\set{{#1},\dots,{#2}}}
\newcommand{\olsi}[1]{\,\overline{\!{#1}}} 
\def\leq{\leqslant}\def\le{\leq}
\def\geq{\geqslant}\def\ge{\geq}

\def\eps{\varepsilon}
\def\PP{\PU_{\sigma,m}}%

\newcommand{\sw}{\mathrm{sewing}}
\newcommand{\PU}{\mathbb{\Pi}}
\newcommand{\cF}{\mathcal{F}}
\newcommand{\sF}{\mathscr{F}}
\newcommand{\sU}{\mathscr{U}}

\newcommand{\mS}{\Sigma}
\newcommand{\sW}{\mathscr{W}}

\newenvironment{myfigure}%
  {\begin{figure}[htb!]\centering}%
  {\end{figure}}

\makeatletter
\let\oldsection\section
\def\mysection#1{\oldsection[\texorpdfstring{#1}{}]{\boldmath{#1}}}
\def\mysectionstar#1{\oldsection*{\boldmath{#1}}}
\def\section{\@ifstar{\mysectionstar}\mysection}
\let\oldsubsection\subsection
\def\mysubsection#1{\oldsubsection[\texorpdfstring{#1}{}]{\boldmath{#1}}}
\def\mysubsectionstar#1{\oldsubsection*{\boldmath{#1}}}
\def\subsection{\@ifstar{\mysubsectionstar}\mysubsection}
\makeatother 

\pgfdeclarelayer{background}
\pgfdeclarelayer{foreground}
\pgfsetlayers{background,main,foreground}

\tikzstyle{polyEdge}=[draw=red,thick, ->]
\tikzstyle{polyVertex}=[fill=white, draw=black, shape=circle, minimum size=5pt, inner sep=0pt]
\tikzstyle{smallNode}=[fill=black, draw=black, shape=circle, inner sep=0pt, minimum size=5pt]
\tikzstyle{gateNode}=[fill=gray, draw=gray, shape=circle, inner sep=0pt, minimum size=5pt]

\tikzstyle{psurface}=[-, draw=gray, fill=gray, opacity=.1]
\tikzstyle{psurfaceLine}=[-, draw=gray]
\tikzstyle{phole}=[-, fill=white, draw=gray]
\tikzstyle{tinyNode}=[fill=black, draw=black, shape=circle, inner sep=0pt, minimum size=2pt]
\tikzstyle{tinygateNode}=[fill=gray, draw=gray, shape=circle, inner sep=0pt, minimum size=2pt]

\title{\textbf{Minor-Universal Graph for Graphs on Surfaces}}

\author{
  Cyril Gavoille\thanks{\textsf{gavoille@labri.fr}. This work is partially funded by the French ANR project ANR-22-CE48-0001 (TEMPOGRAL).}\\
  LaBRI\\
  University of Bordeaux\\
  France
  \and
  Claire Hilaire\thanks{\textsf{claire.hilaire@uca.fr}}\\
  LIMOS\\
  University of Clermont-Auvergne\\
  France
}

\begin{document}

\maketitle

\begin{abstract}
  We show that for integer every $n$ and every surface $\mS$, there is a graph embeddable on $\mS$ with at most $c n^2$ vertices that contains as a minor every $n$-vertex graph embeddable on $\mS$. 
  The constant $c$ depends polynomially on the Euler genus of $\mS$. This generalizes a well-known result for planar graphs by Robertson, Seymour, and Thomas [\textit{Quickly Excluding a Planar Graph.} J. Comb. Theory B, 1994], which states that the square grid on $4n^2$ vertices contains as a minor every $n$-vertex planar graph, an important step in showing that graphs excluding that graphs excluding a planar graph as a minor have bounded tree-width.
  According to Gorsky, Seweryn, and Wiederrecht [\textit{Polynomial Bounds for the Graph Minor Structure Theorem} FOCS~'25], our construction provides the final key ingredient in the search for polynomial bounds in the decomposition of graphs excluding as minor a given $n$-vertex graph achieving tight bounds with respect to the Euler genus of the surface part.

  \medskip 

  \paragraph{Keywords:} Minor, Universal Graphs, Bounded Euler Genus
\end{abstract}

\section{Introduction}
\label{sec:intro}


Graph universality plays an important role in graph theory and in algorithm design. Rather than considering each graph of a given family individually, it is usually simpler to manipulate a small set of graphs --- or even a single graph --- whose properties sufficiently mirror those of the family.


For example, Lipton and Tarjan~\cite{LT80} showed that every $n$-vertex planar graph has a set of $O(\sqrt{n^{2-\alpha}}\,)$ vertices whose removal leaves the graph with connected components of at most $n^{\alpha}$ vertices\footnote{The original statement in~\cite[Theorem~3]{LT80} is formulated differently, as it allows each vertex to have a non-negative weight summing to one. It states that for every $\eps\in[0,1]$, there is a set of $O(\sqrt{n/\eps})$ vertices whose removal leaves components of cost at most $\eps$, which is at most $\eps n$ vertices assuming uniform cost $1/n$. Our formulation is obtained by setting $\eps = 1/n^{1-\alpha}$.}. Choosing $\alpha = 2/3$, this implies that every $n$-vertex planar graph is a subgraph of a graph obtained from a star by blowing up each vertex with a clique of $O(n^{2/3})$ vertices.
This Star-Blowup Theorem has numerous applications for algorithms and data structures for planar graphs. Along the same lines, the recent Fan-Blowup Theorem for planar graphs~\cite{DJMMW25} states that every $n$-vertex planar graph is a subgraph of a fan\footnote{A star whose leaves are connected by a path.} blown up by a clique of $O(\sqrt{n}\log^2{n})$ vertices. Both blowup theorems generalize to Euler genus-$g$ graphs by adding $O(\sqrt{gn}\,)$ vertices to the blowing-up clique.
Another example is the design of NCA-labeling schemes~\cite{AGKR04}, which assign a short label to each vertex of an $n$-vertex tree so that the nearest common ancestor of any two vertices can be inferred solely from their labels. 
Interestingly, Gawrychowski, Kuhn, {\L}opusza{\'n}ski, Panagiotou and Su~\cite{GKLPP18} showed how a specific tree with $n^c$ vertices containing all $n$-vertex trees as a topological minor can be used for the design of succinct and efficient NCA-labeling schemes, that is, with $(c+o(1))\log_2{n}$-bit labels where $c \approx 2.31757$. Similar connections are well-known between adjacency-labeling schemes and induced-universal graphs~\cite{KNR88,ADBTK17,DEGJx21}. Yet another example is the celebrated Grid Minor Theorem of Robertson and Seymour~\cite{RS86a}, which states that graphs excluding a fixed planar graph as a minor have bounded tree-width\footnote{The \emph{tree-width} of $G$ is a measure of how $G$ is close to a tree. It is the smallest $k$ such that $G$ is a subgraph of a chordal graph with maximum clique size $k+1$.}. This result relies on the fact that every planar graph is a minor of a relatively small grid. Thus, proving that the property holds when excluding a single grid suffices to prove that it holds when excluding any planar graph.

\paragraph{Our results.}

Motivated by the latter example, in this paper we focus on universality under minor\footnote{Recall that $H$ is a minor of $G$ if $H$ can be obtained from $G$ by vertex removals, edge removals, and edge contractions.} containment. More precisely, a graph $\sU$ is \emph{minor-universal} for a given graph family $\sF$ if every graph of $\sF$ is a minor of $\sU$. Since we require the graph $\sU$ to be small and to retain key structural properties of the graphs in $\sF$, we typically require $\sU$ and the graphs in $\sF$ to share a common property, such as planarity. Note that the Star-Blowup Theorem~\cite{LT80} and the Fan-Blowup Theorem~\cite{DJMMW25} do not serve this purpose, as the universal graphs of polynomial size they yield for planar graphs are far from planar, containing cliques of $\Omega(\sqrt{n}\,)$ vertices.

A key element of the Grid Minor Theorem~\cite{RS86a} is the fact that every planar graph $H$ is a minor of a grid whose size is bounded by a function of the size of $H$. More precisely, Robertson, Seymour and Thomas~\cite[Theorems~(1.3) \& (1.4)]{RST94} showed that the $(2n\times 2n)$-grid contains as minor every planar graph with $n$ vertices. In other words, the grid on $4n^2$ vertices is a planar minor-universal graph for the $n$-vertex planar graphs. 
This quadratic bound is essentially optimal, as long as grids are concerned, since Biedl, Chambers, Eppstein, de Mesmay, and Ophelders~\cite[Lemma~5]{BCEMO19} showed that any minor-universal grid for planar graphs must have $\Omega(n^2)$ vertices.

We generalize this result of minor-universality for planar graphs 
to graphs on surface of higher genus as follows:

\begin{theorem}\label{th:main}
  For every $n$ and every surface $\mS$ of Euler genus $g\geq 1$, there is a graph $\sU_{n,\mS}$ embedded on $\mS$ with $O(g^2(n+g)^2)$ vertices that contains as minor every graph embeddable on $\mS$ with $n$ vertices.
\end{theorem}

In other words, $\sU_{n,\mS}$ is a minor-universal graph for $n$-vertex graphs embeddable on $\mS$. Moreover, $\sU_{n,\mS}$ has $O(n^2)$ vertices for every surface $\mS$ of constant Euler genus. In fact, as we will see more precisely in \cref{th:tec}, 
the graph $\sU_{n,\mS}$ depends on $\mS$ only through its Euler genus $g$ and its orientability. Moreover, the embedding of the minor is preserved in $\sU_{n,\mS}$, i.e., it is a surface minor of $\sU_{n,\mS}$. 

Following the observation used by Robertson et al.~\cite{RST94} for planar graphs, \cref{th:main} readily extends to graphs on any surface with multiedges and loops. Indeed, by subdividing each edge of the input graph $G$ twice, we obtain a simple graph that contains $G$ as a minor, is embeddable on $\mS$, and has $m = |V(G)| + 2|E(G)|$ vertices. The resulting graph is therefore a minor of $\sU_{m,\mS}$, and so is $G$. Furthermore, the result extends to induced minor containment, since our graph $\sU_{n,\mS}$ has maximum degree at most\footnote{Obviously, we could make the graph sub-cubic by blowing up each vertex by a cycle of length its degree, increasing by a factor at most~$6$ the number of vertices of $\sU_{n,\mS}$.}~$6$. Thus, subdividing each edge yields a graph that contains every $n$-vertex graph embeddable on $\mS$ as an induced minor, with asymptotically the same number of vertices.

Our proof is constructive: both the minor-universal graph $\sU_{n,\mS}$ and the minor witness in $\sU_{n,\mS}$ can be constructed in polynomial time.
As a byproduct, our approach provides an alternative, self-contained proof for the minor-universal grid construction for planar graphs from~\cite{RST94}, without relying on Whitney's Theorem (see \cref{prop:altRST}).

In fact, our result holds in a more general setting. 
As we will see in \cref{sec:prelim}, every surface can be characterized thanks to a signature corresponding to a polygonal schema and a cut-graph of the surface. 
A surface can be associated with several signatures, each corresponding to a polygonal schema, but each signature $\sigma$ defines a unique surface $\mS_{\sigma}$ up to homeomorphism. 
The proof of \cref{th:main} relies on our technical result, which works for any signature (cf. \cref{th:tec} in \cref{sec:main}).
The final step, from \cref{th:tec} to \cref{th:main}, relies on results due to \cite{LPVV01,FHdM24} that bound the number of intersections between a graph embedded on a surface and a cut-graph of that surface associated to its canonical signature.
Therefore, the construction of our minor-universal graph in \cref{th:main} is related to these intersection numbers for canonical signatures (see \cref{lem:crossing_systemloops} in \cref{sec:main}).
As pointed out by Gorsky, Seweryn and Wiederrecht~\cite{GSW25a,GSW25} in their new proof of the Graph Minor Structure Theorem~\cite{RS03}, \cref{lem:crossing_systemloops} can be replaced by a more general one due to \cite{Negami01,FHdM24} that bounds these intersection numbers for any cut-graph, and thus any signature (see \cref{lem:crossing_cutgraph} in \cref{sec:main}).
The resulting bound on the minor-universal graph is weaker, but it holds for every signature $\sigma$, i.e., every surface $\mS_{\sigma}$ with signature $\sigma$ (see \cref{th:main2} below).
This extension is a key to the proof of the polynomial bounds for the Graph Minor Structure Theorem of Gorsky et al.~\cite{GSW25a,GSW25}, according to the authors.


We denote by $\mS_{\sigma}$ the surface obtained from the polygonal schema of signature $\sigma$, which consists of a $2|\sigma|$-regular polygon. See \cref{sec:prelim} for details on these concepts. 

\begin{theorem}\label{th:main2}
  For every $n$ and every signature $\sigma$, there is a graph $\sU_{n,\sigma}$ embedded on $\mS_{\sigma}$ with $O(|\sigma|^4(n+|\sigma|)^2)$ vertices that contains as minor every graph embeddable on $\mS_{\sigma}$ with $n$ vertices.
\end{theorem}

In a sense, \cref{th:main2} is more general than \cref{th:main} as it applies to any signature, not only the canonical one. However, \cref{th:main2} produces minor-universal graphs that are larger by a factor $\Omega(g^2)$ w.r.t. \cref{th:main}, even for the canonical signature.  

\paragraph{Impact on Graph Minor Theory.}

Graph minor Theory centers around two fundamental results: the Grid Minor Theorem (GMT~\cite{RS86a}) and the Graph Minor Structure Theorem (GMST~\cite{RS03}).

As mentioned above, the GMT states that graphs excluding a planar graph as a minor have bounded tree-width. More precisely, every graph of tree-width at least $f(n)$ contains every $n$-vertex planar graph as a minor. The original enormous bound on $f(n)$ has been reduced to $2^{O(n^2)}$ by Robertson, Seymour, and Thomas~\cite{RST94}, eventually to a polynomial one \cite{CT21} with the bound $f(n) \le n^{98 + o(1)}$ (see therein the references for the long sequence of improvements).
Here again, for the GMT it suffices to prove that graphs of tree-width at least $f(n)$ contain as a minor a small minor-universal graph for $n$-vertex planar graphs, namely the $(2n \times 2n)$-grid.
Interestingly, Wollan~\cite[Conjecture~7.3]{Wollan22} proposed an extension of the above GMT statement where ``planar graph'' is replaced by ``graph on surface'', and ``tree-width'' by ``face-width\footnote{Also known as \emph{representativity} introduced by Robertson and Seymour~\cite{RS88}, the face-width measures how densely a graph is embedded on $\mS$. It is the smallest number $k$ such that $\mS$ contains a non-contractible closed curve that intersects the graph in $k$ points, cf.~\cite[Chp.~5]{MT01}.}''.
More precisely, he conjectured that every graph embedded on a surface $\mS$ of Euler genus $g$ with face-width at least some polynomial $p(n,g)$ contains as a minor every $n$-vertex graph embedded on $\mS$. 
His conjecture was confirmed by Gorsky et al.~\cite[Theorem~15.11]{GSW25}, in which our minor-universal graph for graphs on surfaces plays the role played by the grid in the GMT.

The second fundamental theorem, the GMST, describes the structure of graphs excluding a given $n$-vertex graph $H$ as a minor. Specifically, such graphs admit a tree-decomposition whose \emph{torsos}\footnote{They are subgraphs induced by the bags of the tree-decomposition, where each intersection with any other bag is completed by a clique.} are $t(n)$-almost embeddable on surfaces on which $H$ cannot be embedded.
Roughly speaking, a graph is $t(n)$-almost embeddable on $\mS$ if it can be embedded on $\mS$ after deleting $t(n)$ vertices (called \emph{apices}), and deleting certain vertices confined on the boundary of at most $n^2$ faces of the embedding (called \emph{vortices}), which in turn have path-width\footnote{It is the smallest $k$ such that the graph is a subgraph of an interval graph with maximum clique size $k+1$.} at most $t(n)$.
In their original proof, Robertson and Seymour~\cite{RS03} did not provide an explicit bound $t(n)$, and establishing explicit, near-optimal bounds has remained a challenge ever since.
Kawarabayashi, Thomas, and Wollan~\cite{KTW21} made a major step forward in this direction. They provided a new, fully constructive proof of a variant of the GMST, at the expense of losing surface ``tightness''. In their variant, the ``surfaces where $H$ does not embed'' are replaced by ``surfaces of Euler genus $O(n^2)$''.
Thilikos and Wiederrecht~\cite{TW24} overcame this limitation, obtaining a fully constructive proof of the GMST (with surfaces of tight Euler genus) by employing\footnote{Technically they needed to use \cref{th:main2}, which was not online in~\cite{GH23v1} at the time of the writing of their paper. So they made a construction from the minor-universal graph in \cref{th:main} to get a minor-universal graph as the one in \cref{th:main2}.} \cref{th:main}. According to the authors, this was precisely the missing piece that prevented Kawarabayashi et al.~\cite{KTW21} from obtaining tight surfaces.
Unfortunately, their bound on $t(n)$ is still exponential, namely $2^{n^{O(1)}}$. Eventually, still with our construction given by \cref{th:main2}, Gorsky et al. \cite{GSW25a,GSW25} gave a full proof of GMST with polynomial bound of $t(n) = O(n^{2300})$.

\paragraph{Related results on minor-universal graphs.}

For trees, Bodini~\cite{Bodini02} proved that the size of the smallest minor-universal tree for $n$-vertex trees lies between $\Omega(n\log{n})$ and $O(n^{1.985})$.
In fact, the result of Gawrychowski et al.~\cite{GKLPP18} implies that this smallest minor-universal tree has between $\Omega(n^{1.724})$ and $O(n^{1.894})$ vertices\footnote{For these results, rooted binary trees and topological minors were considered. However, this applies to our setting since, up to a factor of two, the rooted and unrooted versions are equivalent, and minor and topological minor containment are equivalent in binary trees.}.
From computational perspectives, the smallest minor-universal tree of two bounded degree trees can be computed in cubic time~\cite{NRT00}, whereas deciding whether a tree is a minor of another is NP-complete. 
Minor universality has also been investigated for infinite graphs~\cite{DK99a,HMSTW21,Lehner22}. 
However, as noticed by Diestel and K{\"u}hn~\cite{DK99a}, the problem already becomes different in the planar case since the infinite grid is no longer minor-universal for infinite planar graphs. The infinite planar graphs having the infinite grid as minor-universal graph have been characterized by K{\"u}hn~\cite{Kuhn01}.

\paragraph{Discussion and technical issues.}

At first glance, the quadratic size of our grid-like construction is not surprising, particularly for surfaces of low genus—such as the real projective plane, the torus, or the Klein bottle, all having Euler genus $g \le 2$.
Indeed, such surfaces can be obtained from a square grid by gluing its sides. 
Thus, embedding the input $n$-vertex graph $G$ into this wrapped grid should be achievable within $O(n^2)$ grid points using standard graph drawing algorithms~\cite{dFPP88,TT89,Schnyder90}, provided these algorithms are suitably adjusted since, for grid minor embeddings, edges must be routed along grid paths (as in orthogonal drawings) rather than straight-line segments of arbitrary slope.

However, there are several technical issues with this approach. First, the relation between the minimum grid-drawing area for a planar $G$ and the size of the minimum grid containing $G$ as minor is not well understood. 
To the best of our knowledge, it is only known that $G$ is a minor of a grid of size $O(A^{3/2})$ if $G$ has a polyline grid-drawing of area $A$, as proved by Dieng and Gavoille~\cite{DG20}. At best, this approach yields an $O(n^3)$ bound for such grid-like graphs when the Euler genus satisfies $g \le 2$.

Secondly, even if a linear relation is conjectured, drawing graphs on surfaces of higher genus ($g \gg 2$) remains challenging.
In this context, drawing a graph of Euler genus $g$ means embedding an unfolded version of the graph without edge crossings, such that its vertices and edges lie on integer grid points inside a convex $2g$-gon representing a polygonal schema of the surface.
For orientable surfaces, Duncan, Goodrich, and Kobourov~\cite{DGK10} proposed an algorithm that may produce straight-line drawing with an area exponential in $n$. 
This exponential bound has been improved to polynomial by Chambers, Eppstein, Goodrich, and L{\"o}ffler~\cite{CEGL12} with $O(g^4 n^6)$ area drawings\footnote{The original result is stated differently, in terms of minimum angular resolution, but this is equivalent.}.
However, we are unaware of any graph drawing results for non-orientable surfaces, even with exponential area bounds. It is important to realize that the edges of an embedded graph may cross many times the sides of the polygonal schema of the surface on which the graph is embedded. One of the key points of drawing algorithms such as those in \cite{DGK10} and \cite{CEGL12} is the design of polygonal schemas whose sides are chord-free paths and with as few vertex repetitions as possible. Until recently, it was only known how to bound the number of side intersections per edge in orientable surfaces. The difficulty in bounding this number of crossings for non-orientable surfaces relies on a conjecture of Negami~\cite{Negami01} on the joint crossing number of two embeddable graphs. The conjecture has been solved recently by Fuladi, Hubard, and de Mesmay~\cite{FHdM22,FHdM24}, in the particular case where one of the two graphs is a one-vertex graph (a loop system), a feature we will use in \cref{lem:crossing_systemloops}.

\paragraph{Organization of the paper.}

In \cref{sec:prelim}, we give the definitions and notations necessary to manipulate graphs on surfaces, with the introduction of \emph{polygonal embedding}, which can be seen as a plane representation of graphs on surface. In \cref{sec:main}, we give our main technical theorem (\cref{th:tec}), which is a minor-universal result for polygonal embeddings.
We then show how this key ingredient implies \cref{th:main,th:main2}, and how its proof relies on two lemmas. Those lemmas are then proved in \cref{sec:proofs}.


\section{Preliminaries}
\label{sec:prelim}

\paragraph{Surfaces.}

In our context, a \emph{surface} is a  compact connected 2-manifolds without boundary.
The Classification Theorem for surfaces (see~\cite[\S6.4]{GX13}) implies that such surfaces can be completely classified, up to homeomorphism, using two topological invariants: \emph{orientability} and \emph{Euler genus}.
A surface is \emph{non-orientable} if traversing a closed loop reverses orientation (changing ``clockwise'' to ``counterclockwise''); otherwise, it is \emph{orientable}.
Its Euler genus is a nonnegative integer that quantifies the complexity of the surface in terms of handles and cross-caps.
It is $2h+c$ if the surface is homeomorphic to a sphere augmented with $h$ handles and $c$ cross-caps.
Adding a handle to a sphere consists in removing two disjoint open disks of the sphere and in gluing to their boundaries a cylinder by its two cycle boundaries.
Adding a cross-caps consists in removing one disk of the sphere and gluing a Möbius strip by its boundary. 
The surface is non-orientable if and only if it contains at least a cross-cap.
In fact, by the Classification Theorem, any Euler genus-$g$ surface can be obtained by adding to a sphere $g/2$ handles (if orientable) or by $g$ cross-caps (if non-orientable).
The sphere and the torus are examples of orientable surfaces of Euler genus~0 and~2, respectively, whereas the real projective plane and the Klein bottle are non-orientable surfaces of Euler genus~1 and~2, respectively. We refer the reader to \cite{Stillwell93,GX13,CdV21} for more precise definitions about surfaces.

\paragraph{Polygonal schema.}

Consider any surface $\mS$. There exist infinitely many polygons from which $\mS$ can be constructed, up to homeomorphism, by gluing their sides in pairs. For instance, the sphere is homeomorphic to a pyramid or a cube that can be obtained by gluing the sides of a hexagon and a 14-gon respectively. For any such polygon~$P$ (so with an even number of sides), the way the pairs of sides of~$P$ are glued together to form $\mS$ can be described by a \emph{signature}, a word associating clockwise one symbol (either positive or negative) with each side of $P$. 
Each side of $P$ is oriented either clockwise if its symbol is positive, or counter-clockwise otherwise. 
Two sides with the same (or opposite) symbol in the signature are glued together with respect to their orientation, creating a twist (or cross-cap) in the surface if they are opposite, and simply a handle if they are not opposite.
The polygon $P$ with its signature form a \emph{polygonal schema} of $\mS$.
After this gluing process, sides and vertices of $P$ form a graph (potentially with multi-edges and loops) embedded on $\mS$ without edge crossings, called \emph{cut-graph} of $\mS$. Obviously, cutting $\mS$ along the edges of this cut-graph results in $P$, each edge corresponding to two sides of~$P$. More generally, a cut-graph of $\mS$ is any graph embedded on $\mS$ such that its cut results into a surface homeomorphic to a disc, and thus to some polygon.

A \emph{canonical system of loops} for $\mS$ is a specific cut-graph with one vertex and~$g$ loops, where~$g$ is the Euler genus of $\mS$. Moreover, loops are either all two-sided if $\mS$ is orientable, or all one-sided if $\mS$ is non-orientable.
For $g \ge 1$, the \emph{canonical signature}, that is, the signature obtained by cutting $\mS$ along this canonical system of loops, is $a_1 a_2 \olsi{a}_1 \olsi{a}_2 \dots a_{g-1} a_{g} \olsi{a}_{g-1} \olsi{a}_g$ if $\mS$ is orientable, and $a_1 a_1 \dots a_g a_g$ if $\mS$ is non-orientable. 
For the sphere, $g = 0$, the signature is\footnote{Sometimes this is fixed to $a_0\olsi{a}_0$.} $\eps$, the empty word. The canonical signature is always composed of exactly~$2g$ letters for all Euler genus $g\in \mathbb{N}$, whereas arbitrary signatures contain at least~$2g$ letters.

\paragraph{Polygonal embeddings.}

In order to adapt the construction of graphs on surfaces in the polygonal schema setting, we introduce the notion of polygonal embedding.
Consider a surface $\mS$ with a cut-graph $L$, and a graph $G$ embedded\footnote{We only consider simple graphs and cellular embeddings, i.e., drawing of the graph on the surface without edge crossings such that each face is homeomorphic to a disc.} on $\mS$.
Denote by $G\uplus L$ the graph embedded on $\mS$ obtained from the union of $G$ and $L$, and by adding a new vertex at each intersection point between $G$ and $L$.
A \emph{polygonal embedding} of $G$ is the planar embedding $\Pi$ of $G$ and $L$ in the polygon $P$ obtained from $G \uplus L$ by cutting $\mS$ along the edges of $L$. (We refer to \cref{fig:schema} for an illustration.) 
The edges and vertices on $L$ appear duplicated on the boundary of the outerface of $\Pi$. More precisely, the boundary of the outerface of $\Pi$ is a cycle that can be cut into a clockwise sequence of paths sharing their extremities. 
These paths, called \emph{sides} of $\Pi$, correspond to the sides of the polygon $P$, and their extremities, called \emph{corners} of $\Pi$, corresponds to the vertices of $P$. Note that corners have degree two in $\Pi$. 
The ordered sequence of corners is called the \emph{border} of $\Pi$. The \emph{signature} of $\Pi$, denoted by $\sigma(\Pi)$, corresponds to the signature of $P$ and describes how the sides of $\Pi$ are reattached to form $G\uplus L$ on $\mS$. More precisely, the $i$th symbol of the signature is associated with the $i$th side of $\Pi$, i.e., the path between the $i$th and $(i+1)$th corners.
These sides are merged according to the orientation given by the symbols, and they are called \emph{twin sides} of $\Pi$. Obviously, twin sides must contain the same number of vertices. Furthermore, the vertices and edges that are identified in this process are called respectively \emph{twin vertices} and \emph{twin edges}.

\begin{myfigure}
  \begin{subfigure}[b]{0.55\textwidth}
      \begin{tikzpicture}[scale=0.5]
        %
        
        \begin{pgfonlayer}{foreground}
        \node (19) at (-7, 0) {};
        \node (20) at (7, 0) {};
        \node (21) at (-2.25, 0.25) {};
        \node (22) at (2.25, 0.25) {};
        \node (23) at (2, 0) {};
        \node (24) at (-2, 0) {};
        \node (30) at (1.5, -0.25) {};
        \node (31) at (4, -3.5) {};
        \node (32) at (-3, 0) {};
        \node (33) at (3, 0) {};
        \node [style=polyVertex] (0) at (1.75, -1) {};
        \node (70) at (1.5, -1.5) { \textcolor{black}{$r$}};
        \node (47) at (0, 1.8) { \textcolor{blue}{$\ell_1$}};
        \node (48) at (2.2, -2.6) { \textcolor{red}{$\ell_2$}};
        \node [style=smallNode] (1) at (-2.25, -1.5) {};
        \node [style=smallNode] (2) at (-1.75, -2.75) {};
        \node [style=smallNode] (3) at (-3.75, -2) {};
        \node [style=smallNode] (4) at (-0.5, -2) {};
        \node [style=smallNode] (5) at (-3.75, -1) {};
        \node [style=smallNode] (6) at (-3, -3) {};
        \node (55) at (-2, 0) {};
        \node (56) at (-1.5, -0.25) {};
        \node (57) at (-0.75, -0.5) {};
        \node (58) at (-4, -3.5) {};
        \node (59) at (-3, -3.75) {};
        \node (60) at (-1.75, -4) {};
        \node (61) at (0.5, -0.5) {};
        \node (62) at (1.25, -4) {};
        \node (64) at (4, 0) {};
        \node (65) at (5, 0) {};
        \node (66) at (6, 0) {};
        \node (67) at (-4, 0) {};
        \node (68) at (-5, 0) {};
        \node (69) at (-6, 0) {};
        \end{pgfonlayer}
        \draw [style=psurface, bend right=90] (19.center) to (20.center) to (19.center);
        \draw [style=phole] (24.center) to[in=225, out=-45, looseness=0.75] (23.center) to[bend right=90, looseness=0.50] (24.center);
        \draw [style=psurfaceLine, bend right=60, looseness=0.75] (21.center) to (22.center);
        \draw [style=polyEdge, -, dashed, in=15, out=-15, looseness=0.50] (31.center) to (30.center);
        \draw [style=polyEdge, -, bend right=15, looseness=0.75] (0) to (31.center);
        \draw [style=polyEdge, in=130, out=-150] (30.center) to (0);
        \draw [style=polyEdge, color=blue, in=-105, out=0, looseness=0.75] (0) to (33.center);
        \draw [style=polyEdge, color=blue, bend left=270, looseness=0.75] (33.center) to (32.center);
        \draw [style=polyEdge, color=blue, in=-165, out=-75, looseness=0.75] (32.center) to (0);
        \draw (3) to (1);
        \draw (1) to (2);
        \draw (2) to (3);
        \draw (3) to (6);
        \draw (6) to (2);
        \draw (2) to (4);
        \draw (4) to (1);
        \draw (1) to (5);
        \draw (5) to (3);
        \draw [in=90, out=0, looseness=0.50] (57.center) to (4);
        \draw [in=75, out=150, looseness=0.50] (56.center) to (1);
        \draw [in=45, out=135, looseness=0.75] (55.center) to (5);
        \draw [in=-180, out=-45, looseness=0.75] (6) to (60.center);
        \draw [in=-180, out=-90, looseness=0.75] (6) to (59.center);
        \draw [in=-15, out=-150, looseness=0.75] (6) to (58.center);
        \draw [dashed, in=-15, out=-180, looseness=0.50] (57.center) to (60.center);
        \draw [dashed, in=-60, out=-15, looseness=0.25] (56.center) to (59.center);
        \draw [dashed, in=150, out=-75, looseness=0.50] (55.center) to (58.center);
        \draw [in=165, out=60, looseness=0.75] (4) to (61.center);
        \draw [dashed, in=-135, out=-15, looseness=0.50] (61.center) to (62.center);
        \draw [in=-90, out=0, looseness=0.75] (4) to (64.center);
        \draw [bend left=270] (64.center) to (67.center);
        \draw [bend right=15] (67.center) to (5);
        \draw [in=-90, out=-15, looseness=0.75] (2) to (65.center);
        \draw [bend left=270] (65.center) to (68.center);
        \draw [in=165, out=-90] (68.center) to (5);
        \draw [bend left=270] (66.center) to (69.center);
        \draw [in=165, out=-90] (69.center) to (3);
        \draw [in=30, out=-90, looseness=0.75] (66.center) to (62.center);
        
      \end{tikzpicture}
  \end{subfigure}
  \hfill
  \begin{subfigure}[b]{0.4\textwidth}
    \begin{tikzpicture}[scale=0.6]
      \draw [style=psurface]  (-4, -3) rectangle (4, 3);
      \node [style=polyVertex] (c1) at (-4, -3) {};
      \node [style=polyVertex] (c2) at (-4, 3) {};
      \node [style=polyVertex] (c3) at (4, 3) {};
      \node [style=polyVertex] (c4) at (4, -3) {};
      \draw [style=polyEdge] (c3) -- (c4) node[midway,right=1pt]{ \textcolor{red}{$a_2$}};
      \draw [style=polyEdge] (c2) -- (c1)node[midway,left=1pt]{ \textcolor{red}{$\bar{a_2}$}};
      \draw [style=polyEdge, color=blue] (c2) -- (c3)node[midway,above=1pt]{ \textcolor{blue}{$a_1$}};
      \draw [style=polyEdge, color=blue] (c1) -- (c4)node[midway,below=1pt]{ \textcolor{blue}{$\bar{a_1}$}};
      \node [style=smallNode] (1) at (0, 1.5) {};
      \node [style=smallNode] (2) at (1, 0) {};
      \node [style=smallNode] (3) at (-1, 0) {};
      \node [style=smallNode] (4) at (2, 1.5) {};
      \node [style=smallNode] (5) at (0, -1.5) {};
      \node [style=smallNode] (6) at (-2, 1.5) {};
      \node [style=gateNode] (a11) at (-1.5, 3) {};
      \node [style=gateNode] (a12) at (0, 3) {};
      \node [style=gateNode] (a13) at (1.5, 3) {};
      \node [style=gateNode] (a14) at (3, 3) {};
      \node [style=gateNode] (b11) at (-1.5, -3) {};
      \node [style=gateNode] (b12) at (0, -3) {};
      \node [style=gateNode] (b13) at (1.5, -3) {};
      \node [style=gateNode] (b14) at (3, -3) {};
      \node [style=gateNode] (a21) at (4, 1.5) {};
      \node [style=gateNode] (a22) at (4, 0.75) {};
      \node [style=gateNode] (a23) at (4, -1) {};
      \node [style=gateNode] (b21) at (-4, 1.5) {};
      \node [style=gateNode] (b22) at (-4, 0.75) {};
      \node [style=gateNode] (b23) at (-4, -1) {};
      \draw (6) to (3);
      \draw (3) to (5);
      \draw (5) to (2);
      \draw (2) to (4);
      \draw (4) to (1);
      \draw (1) to (6);
      \draw (3) to (1);
      \draw (1) to (2);
      \draw (2) to (3);
      \draw (a11) to (6);
      \draw (b11) to (5);
      \draw (1) to (a12);
      \draw (b12) to (5);
      \draw (4) to (a13);
      \draw (b13) to (5);
      \draw [bend right=15, looseness=0.75] (4) to (a14);
      \draw (4) to (a21);
      \draw (2) to (a22);
      \draw (b21) to (6);
      \draw (b22) to (6);
      \draw [bend left] (b14) to (a23);
      \draw [bend left=15, looseness=0.75] (b23) to (3);
    \end{tikzpicture}
  \end{subfigure}
  \caption{On the left, the graph $K_6$ embedded on the torus (of Euler genus~2)
    with the canonical system of loops $L = (\set{r}, \set{\ell_1,\ell_2})$. On the right, a polygonal embedding $\Pi$ of $K_6$ of signature
    $\sigma = a_1a_2\bar{a_1}\bar{a_2}$. By construction, $K_6$ is a
    minor of the graph $\sw(\Pi) = K_6 \uplus L$.
  }
  \label{fig:schema}
\end{myfigure}

Given a polygonal embedding $\Pi$, the \emph{sewing} of $\Pi$, denoted by $\sw(\Pi)$, is the graph obtained by reattaching the sides of $\Pi$ according to its border and its signature. It is clear that, if $\Pi$ comes from $G\uplus L$ embedded on $\mS$ by cutting along the edges of $L$, then $\sw(\Pi)$ is isomorphic to $G\uplus L$. In particular, $\sw(\Pi)$ embeds on $\mS$ and contains $G$ as minor. In fact, the vertices of the cut-graph can always be chosen to not belong to $G$, and thus $\sw(\Pi)\setminus V(L)$ contains $G$ as minor too.

Observe that, given a polygonal embedding $\Pi$, the simple graph obtained by adding an edge between each pair of twin vertices, removing the corner vertices admits $\sw(\Pi)\setminus V(L)$ as minor. 
Indeed, contracting the edges between the twins in this simple graph is equivalent in merging the twins in $\Pi$. See \cref{fig:simple_graph}.

\begin{myfigure}
  \begin{subfigure}[b]{0.5\textwidth}
      \begin{tikzpicture}[scale=0.6]
        \begin{pgfonlayer}{foreground}
          \node [style=smallNode] (1) at (0, 1.5) {};
          \node [style=smallNode] (2) at (1, 0) {};
          \node [style=smallNode] (3) at (-1, 0) {};
          \node [style=smallNode] (4) at (2, 1.5) {};
          \node [style=smallNode] (5) at (0, -1.5) {};
          \node [style=smallNode] (6) at (-2, 1.5) {};
          \node [style=gateNode] (a11) at (-1.5, 3) {};
          \node [style=gateNode] (a12) at (0, 3) {};
          \node [style=gateNode] (a13) at (1.5, 3) {};
          \node [style=gateNode] (a14) at (3, 3) {};
          \node [style=gateNode] (b11) at (-1.5, -3) {};
          \node [style=gateNode] (b12) at (0, -3) {};
          \node [style=gateNode] (b13) at (1.5, -3) {};
          \node [style=gateNode] (b14) at (3, -3) {};
          \node [style=gateNode] (a21) at (4, 1.5) {};
          \node [style=gateNode] (a22) at (4, 0.75) {};
          \node [style=gateNode] (a23) at (4, -1) {};
          \node [style=gateNode] (b21) at (-4, 1.5) {};
          \node [style=gateNode] (b22) at (-4, 0.75) {};
          \node [style=gateNode] (b23) at (-4, -1) {};
          \node (d3) at (-4, -4) {};
          \node (dd3) at (4, -4) {};
          \node (d2) at (-4, -4.5) {};
          \node (dd2) at (4, -4.5) {};
          \node (d1) at (-4, -5) {};
          \node (dd1) at (4, -5) {};
        \end{pgfonlayer}
        \draw [color=red, bend right=90] (b23) to (d3) to[bend right=0] (dd3) to (a23) ;
        \draw [color=red, bend right=90] (b22) to (d2) to[bend right=0] (dd2) to (a22);
        \draw [color=red, bend right=90] (b21) to (d1) to[bend right=0] (dd1) to (a21) ;
        \draw [color=blue, bend left=160] (a11) to (b11);
        \draw [color=blue, bend left=160] (a12) to (b12);
        \draw [color=blue, bend left=160] (a13) to (b13);
        \draw [color=blue, bend left=160] (a14) to (b14);
        \draw (a11) to (a14);
        \draw (b11) to (b14);
        \draw (a21) to (a23);
        \draw (b21) to (b23);
        \draw (6) to (3);
        \draw (3) to (5);
        \draw (5) to (2);
        \draw (2) to (4);
        \draw (4) to (1);
        \draw (1) to (6);
        \draw (3) to (1);
        \draw (1) to (2);
        \draw (2) to (3);
        \draw (a11) to (6);
        \draw (b11) to (5);
        \draw (1) to (a12);
        \draw (b12) to (5);
        \draw (4) to (a13);
        \draw (b13) to (5);
        \draw [bend right=15, looseness=0.75] (4) to (a14);
        \draw (4) to (a21);
        \draw (2) to (a22);
        \draw (b21) to (6);
        \draw (b22) to (6);
        \draw [bend left] (b14) to (a23);
        \draw [bend left=15, looseness=0.75] (b23) to (3);
      \end{tikzpicture}
  \end{subfigure}
  \hfill
  \begin{subfigure}[b]{0.45\textwidth}
      \begin{tikzpicture}[scale=0.5]
        \begin{pgfonlayer}{foreground}
        \node (19) at (-7, 0) {};
        \node (20) at (7, 0) {};
        \node (21) at (-2.25, 0.25) {};
        \node (22) at (2.25, 0.25) {};
        \node (23) at (2, 0) {};
        \node (24) at (-2, 0) {};
        \node (30) at (1.5, -0.25) {};
        \node (31) at (4, -3.5) {};
        \node (32) at (-3, 0) {};
        \node (33) at (3, 0) {};
        %
        \node [style=gateNode] (a21) at (2.1, -1.55) {};
        \node [style=gateNode] (a22) at (2.65, -2.3) {};
        \node [style=gateNode] (a23) at (3.25, -2.9) {};
        \node [style=gateNode] (a11) at (-2.9, -0.2) {};
        \node [style=gateNode] (a12) at (-2, -0.8) {};
        \node [style=gateNode] (a13) at (-0.55, -1.2) {};
        \node [style=gateNode] (a14) at (-0.1, -1.2) {};
        \node [style=smallNode] (1) at (-2.25, -1.5) {};
        \node [style=smallNode] (2) at (-1.75, -2.75) {};
        \node [style=smallNode] (3) at (-3.75, -2) {};
        \node [style=smallNode] (4) at (-0.5, -2) {};
        \node [style=smallNode] (5) at (-3.75, -1) {};
        \node [style=smallNode] (6) at (-3, -3) {};
        \node (55) at (-2, 0) {};
        \node (56) at (-1.5, -0.25) {};
        \node (57) at (-0.75, -0.5) {};
        \node (58) at (-4, -3.5) {};
        \node (59) at (-3, -3.75) {};
        \node (60) at (-1.75, -4) {};
        \node (61) at (0.5, -0.5) {};
        \node (62) at (1.25, -4) {};
        \node (64) at (4, 0) {};
        \node (65) at (5, 0) {};
        \node (66) at (6, 0) {};
        \node (67) at (-4, 0) {};
        \node (68) at (-5, 0) {};
        \node (69) at (-6, 0) {};
        \end{pgfonlayer}
        \draw [style=psurface, bend right=90] (19.center) to (20.center) to (19.center);
        \draw [style=phole] (24.center) to[in=225, out=-45, looseness=0.75] (23.center) to[bend right=90, looseness=0.50] (24.center);
        \draw [style=psurfaceLine, bend right=60, looseness=0.75] (21.center) to (22.center);
        \draw [ bend right=15, looseness=0.75] (a21.center) to (a23.center);
        \draw [in=-180, out=-50, looseness=0.75] (a11.center) to (a14.center);

        \draw (3) to (1);
        \draw (1) to (2);
        \draw (2) to (3);
        \draw (3) to (6);
        \draw (6) to (2);
        \draw (2) to (4);
        \draw (4) to (1);
        \draw (1) to (5);
        \draw (5) to (3);
        \draw [in=90, out=0, looseness=0.50] (57.center) to (4);
        \draw [in=75, out=150, looseness=0.50] (56.center) to (1);
        \draw [in=45, out=135, looseness=0.75] (55.center) to (5);
        \draw [in=-180, out=-45, looseness=0.75] (6) to (60.center);
        \draw [in=-180, out=-90, looseness=0.75] (6) to (59.center);
        \draw [in=-15, out=-150, looseness=0.75] (6) to (58.center);
        \draw [dashed, in=-15, out=-180, looseness=0.50] (57.center) to (60.center);
        \draw [dashed, in=-60, out=-15, looseness=0.25] (56.center) to (59.center);
        \draw [dashed, in=150, out=-75, looseness=0.50] (55.center) to (58.center);
        \draw [in=165, out=60, looseness=0.75] (4) to (61.center);
        \draw [dashed, in=-135, out=-15, looseness=0.50] (61.center) to (62.center);
        \draw [in=-90, out=0, looseness=0.75] (4) to (64.center);
        \draw [bend left=270] (64.center) to (67.center);
        \draw [bend right=15] (67.center) to (5);
        \draw [in=-90, out=-15, looseness=0.75] (2) to (65.center);
        \draw [bend left=270] (65.center) to (68.center);
        \draw [in=165, out=-90] (68.center) to (5);
        \draw [bend left=270] (66.center) to (69.center);
        \draw [in=165, out=-90] (69.center) to (3);
        \draw [in=30, out=-90, looseness=0.75] (66.center) to (62.center);
      \end{tikzpicture}
  \end{subfigure}
  \caption{The graph on the left is constructed from the polygonal embedding $\Pi$ of $K_6$ from \cref{fig:schema}, by removing the four corner vertices and making the twin vertices adjacent. Contracting the colored edges results in the graph $\sw(\Pi)\setminus \{r\}$ embedded on the torus (on the right), which contains $K_6$ as minor.}
  \label{fig:simple_graph}
\end{myfigure}

In order to formalize minor containment preserving signature for polygonal embeddings, we introduce the following relation, called \emph{p-minor} (for polygonal-embedding-minor). More precisely, a polygonal embedding $\Pi$ is a \emph{p-minor} of a polygonal embedding $\Pi'$ if it has the same signature, same corners, and if $\sw(\Pi)$ is a minor of $\sw(\Pi')$.
Moreover, we say that $\Pi$ has \emph{size} $(m,n)$ if it has at most $n$ internal vertices (those that are not lying on the boundary of the outerface of $\Pi$), and each side has at most $m$ vertices, the corners excluded.
In particular, every polygonal embedding $\Pi$ of size $(m,n)$ has at most $n + |\sigma(\Pi)|(m+1)$ vertices\footnote{$|\sigma(\Pi)|$ denotes the total number of symbols in the word $\sigma(\Pi)$.}, and $\sw(\Pi)$ has at most $n + |\sigma(\Pi)|(m+1)/2$ vertices (since vertices on the boundary of the outerface are duplicated). 
Note that a polygonal embedding of size of $(m,0)$ is a Hamiltonian outerplane\footnote{A graph is Hamiltonian if there is a simple cycle spanning all its vertices. An outerplane graph is a plane graph such that all the vertices are on the boundary of the outerface.} graph.


\section{Main results}\label{sec:main}

Recall that our main results (\cref{th:main,th:main2}) state that, given a positive integer $n$ and a surface $\mS$ with signature $\sigma$, one can construct a graph $\sU_{n,\mS}$ embedded on $\mS$ that is minor-universal for all graphs on at most $n$ vertices embeddable on $\mS$. Moreover, the number of vertices of $\sU_{n,\mS}$ is bounded by a polynomial function of $n$ and $|\sigma|$.
If $\sigma$ is given as an input (the surface is then determined by $\sigma$), then \cref{th:main2} yields the bound $O(|\sigma|^4(n+|\sigma|)^2)$. If $\sigma$ is not given, then we choose $\sigma$ to be the canonical signature of $\mS$ ($\sigma$ is then determined by the Euler genus $g$ and the orientability of $\mS$), and \cref{th:main} yields the bound $O(g^2(n+g)^2)$.

In both cases, our minor-universal graphs are defined via the sewing of specific polygonal embeddings, denoted by $\PP$, where $m$ is a positive integer.
The signature of $\PP$ is precisely $\sigma$ and its size is  related to $m$. The polygonal embedding $\PP$ is obtained from the half-$(|\sigma|m\times |\sigma|m)$-grid cut by its diagonal, plus $|\sigma|$ vertices for the corners. 
The non-corner vertices on the boundary of the outerface are those of the grid diagonal, in the natural order, and each side contains exactly $m$ non-corner vertices. See \cref{fig:PU} for an example, and \cref{subsec:proofLemFinal} for a more formal description of $\PP$.
Observe that the polygonal embedding $\PP$ is defined for every signature~$\sigma$ (canonical or not) and for every positive integer~$m$.

\begin{myfigure}
  \begin{tikzpicture}[scale=0.5,rotate=-45]
    \def\k{4}
    \def\m{3}
    \pgfmathsetmacro{\n}{\k*\m-1}
    \pgfmathsetmacro{\km}{\k-1}
    \begin{pgfonlayer}{foreground}
      \foreach \i in {0,...,\n}{ %
        \draw (\i,0) -- (\i,\i) -- (\n,\i);
        \foreach \j in {0,...,\i} \node[tinyNode] at (\i,\j) {};
        \node[tinygateNode] at (\i,\i) {};
      }
    \end{pgfonlayer}
      \foreach \i in {0,...,\km}{
        \pgfmathparse{\m*\i-0.5}
        \node[polyVertex] (u\i) at (\pgfmathresult,\pgfmathresult) {};
      }
      \draw[polyEdge, draw=blue] (u0) -- (u1) node[midway,left, above]{ \textcolor{blue}{$a_1$}};
      \draw[polyEdge] (u1) -- (u2) node[midway,left, above]{ \textcolor{red}{$a_2$}};
      \draw[polyEdge, draw=blue] (u3) -- (u2) node[midway,left, above]{ \textcolor{blue}{$\bar{a_1}$}};
      \pgfmathsetmacro{\a}{\m*\k -0.5}
      \draw[polyEdge] (u0) to[out=-135,in=-135,-] (\a +1,-1.5) to[out=45,in=45,-] (\a,\a) to node[midway,left, above]{ \textcolor{red}{$\bar{a_2}$}} (u3); 
      \draw[<->] (-0.75,1) -- (\m-1.75,\m) node[midway,left, above]{ $m=\m$};
      \draw[<->] (0,-1) -- (\n,-1) node[midway,left=5pt, below]{ $|\sigma|m$};
  \end{tikzpicture}
  \caption{The polygonal embedding $\PP$ with signature $\sigma = a_1a_2\bar{a_1}\bar{a_2}$ (the torus) and $m=3$. The four white vertices are the corners of $\PP$.}
  \label{fig:PU}
\end{myfigure}

We first make the following observation about the size of $\PP$ and the number of vertices of $\sw(\PP)$.

\begin{proposition}\label{prop:PU}
  The embedding $\PP$ is a polygonal embedding of signature $\sigma$ with size $(m,\frac{1}{2}|\sigma|m\cdot(|\sigma|m-1))$. In particular, $\sw(\PP)$ has no more than $\frac{1}{2}(|\sigma|^2m^2 + |\sigma|)$ vertices, and at most $\frac{1}{2}|\sigma|^2m^2 + 1$ vertices if $\sigma$ is canonical.
\end{proposition}

\begin{proof}
The number of internal vertices of $\PP$ is $n = 1 + 2 + \cdots + |\sigma|m-1=\frac{1}{2}|\sigma|m\cdot(|\sigma|m-1))$. 
By merging the twin sides of $\PP$ to obtain $\sw(\PP)$, we destroy at least half of the vertices of the cycle outerface of $\PP$. It follows that $\sw(\PP)$ has no more than $n + |\sigma|(m+1)/2$ vertices. 
However, depending on the signature $\sigma$, this number can be lower as the set of corners may collapse even more when merging the sides. If the signature is canonical, i.e., comes from a canonical system of loops, $\sw(\PP)$ has only $n + |\sigma|m/2 + 1$ vertices since such a signature tell us that all the corners are twins, resulting into the single vertex of the system of loops. 
\end{proof}

Our technical theorem is the following. Informally, this is the analog of our minor-universality theorems in the setting of polygonal embeddings. Note that it applies to any polygonal embedding, not only those whose signature comes from a canonical system of loops, which is a key point for proving \cref{th:main2}. 

\begin{theorem}\label{th:tec}
  Every polygonal embedding $\Pi$ of size $(m,n)$ is a p-minor of $\PU_{\sigma(\Pi),m+2n}$. 
\end{theorem}

\begin{proofof}{\textbf{\cref{th:main,th:main2}.}}
Let us now show how to prove \cref{th:main,th:main2} using \cref{th:tec}.
Consider a surface $\mS$ of Euler genus $g$ and the signature $\sigma$ of any polygonal schema of $\mS$. Recall that this signature corresponds to some cut-graph $L$ of $\mS$ that embeds on $\mS$.
In \cref{th:main2}, the signature $\sigma$ is arbitrary, whereas in \cref{th:main}, $\sigma$ is chosen as the canonical signature.
In the case of a canonical signature, the signature and its cut-graph (a canonical system of loops) are completely determined by the Euler genus and orientability of $\mS$ (see \cref{sec:prelim}). In particular, $|\sigma| = 2g$.

Now consider any $n$-vertex graph $G$ embedded on $\mS$.
Then, for every embedding of $L$ on $\mS$ with $V(L)$ disjoint from $V(G)$, $G$ has a polygonal embedding $\Pi_G$ of signature $\sigma$ such that $G$ is a minor of $\sw(\Pi_G)\setminus V(L)=(G\uplus L)\setminus V(L)$, a graph that embeds on $\mS$ by construction.
Note that $\Pi_G$ has at most $n$ internal vertices due to $G$, and a certain number of vertices lying on the sides of $\Pi_G$ due to $V(L)$ and to the intersections between $G$ and $L$.
To apply \cref{th:tec}, we need to determine the size of $\Pi_G$, and thus bound the maximum number of intersections between an edge of $L$ and the edges of $G$, which we denote by $m$.
Indeed, to bound the size of $\Pi_G$, we require the number of intersections per edge of $L$ (recall that each edge of $L$ corresponds to a side of $\Pi_G$), as the total number of intersections between $G$ and $L$ is not sufficiently precise for our purpose.
We can bound this number thanks to the following results\footnote{In \cref{lem:crossing_cutgraph}, the proof of \cite{Negami01} has an issue for the non-oriented case that has been corrected by~\cite{FHdM24}. In \cref{lem:crossing_systemloops}, we refer to~\cite{CdV21}[Theorem~8.1] for a reformulation of the original statement of~\cite{LPVV01}.}:

\begin{lemma}[\cite{Negami01,FHdM24}]\label{lem:crossing_cutgraph}
  Given an embedding of $G$ on a surface $\mS$ of Euler genus $g$, for every cut-graph $L$ of $\mS$, there is an embedding of $L$ on $\mS$ such that each edge of $L$ intersects any edge of $G$ in at most $O(g)$ points.
\end{lemma}

In the case where we choose the cut-graph to be a canonical system of loops, we can get more precise bounds, as well as a polynomial time algorithm finding such system of loops:

\begin{lemma}[\cite{LPVV01,FHdM24}]\label{lem:crossing_systemloops}
  Given an embedding of $G$ on a surface $\mS$ of Euler genus $g$, there is a polynomial time algorithm that computes a canonical system of $g$ loops such that each loop intersects any edge of $G$ in at most $4$ points if $\mS$ is orientable~\cite{LPVV01}, and at most $30$ points if $\mS$ is non-orientable~\cite{FHdM24}.
\end{lemma}

Note that for both lemmas, we can always make sure that the embedding of $L$ satisfies that $V(L)$ and $V(G)$ are disjoint, perhaps by slightly shifting the embedding of $L$.
By choosing the embedding of $L$ using the previous lemmas, together with the fact that $G$ has at most $3n + 3g - 6$ edges\footnote{From Euler's Formula in simple connected graphs with Euler genus $g$, $n$ vertices, $e$ edges and $f$ faces, that is $n - e + f = \chi(\mS) = 2 - g$, and from the fact that $3f \ge 2e$.}, and the fact that $|\sigma|\ge 2g$, we get that $m = O(g(n+g)) = O(|\sigma|(n+|\sigma|))$ for any $\sigma$ (using \cref{lem:crossing_cutgraph}), and $m = O(n+g)$ if we use a canonical signature (using \cref{lem:crossing_systemloops}).
Note that from the definition of $m$, $\Pi_G$ has size $(m,n)$.  

By \cref{th:tec}, $\Pi_G$ is a p-minor of $\PU_{\sigma,m+2n}$, which by the definition of p-minor implies that $\sw(\Pi_G)$ is a minor\footnote{In fact a surface minor.} of $\sw(\PU_{\sigma,m+2n})$. 
It follows that $G$ is a minor of $\sw(\PU_{\sigma,m+2n})\setminus V(L)$. Consequently, we can define our minor-universal graph (which is independent of $G$) as $\sU_{n,\mS} = \sw(\PU_{\sigma,m+2n})\setminus V(L)$. \cref{fig:PU2} shows several examples of "simple" graphs that contain $\sU_{n,\mS}$ as minor and with a proportional number of vertices.

It remains to show that $\sU_{n,\mS}$ is indeed our minor-universal graph of desired size. 
By construction, it embeds on $\mS$ and contains every $n$-vertex graph that embeds in $\mS$ as a minor.
Finally, $|V(\sU_{n,\mS})| < |V(\sw(\PU_{\sigma,m+2n}))| = O({|\sigma|}^2 m^2)$ by \cref{prop:PU}. Depending on the lemma we used to bound $m$, this gives $O(|\sigma|^4(n+|\sigma|)^2)$ in the general case (\cref{th:main2}) and $O(g^2(n+g)^2)$ for the canonical case (\cref{th:main}).

This completes the proofs of \cref{th:main,th:main2}.
\end{proofof}

\begin{myfigure}
  \begin{subfigure}[b]{0.5\textwidth}
  \centering
    \begin{tikzpicture}[scale=0.45,rotate=-45]
      \def\k{4}
      \def\m{3}
      \pgfmathsetmacro{\n}{\k*\m-1}
      \pgfmathsetmacro{\mminus}{\m-1}
      \pgfmathsetmacro{\kminus}{\k-1}
    \begin{pgfonlayer}{foreground}
      \foreach \i in {0,...,\n}{ %
        \draw (\i,0) -- (\i,\i) -- (\n,\i);
        \foreach \j in {0,...,\i} \node[tinyNode] at (\i,\j) {};
        \node[tinygateNode] at (\i,\i) {};
      }
    \end{pgfonlayer}
        \foreach \i in {0,...,\mminus}{
          \pgfmathparse{\m*3-\i-1}
          \node (a\i) at (\pgfmathresult,\pgfmathresult) {};
          \node (b\i) at (\pgfmathresult+\m,\pgfmathresult+\m) {};
          \draw[polyEdge, draw=blue, out=135, in=135, -] (\i,\i) to (a\i.center);
          \draw[polyEdge, out=135, in=135, -] (\i+\m,\i+\m) to (b\i.center);
        }
        \foreach \i in {0,...,\kminus}{
          \pgfmathparse{(\i+1)*\m-1}
          \node (t\i) at (\pgfmathresult,\pgfmathresult) {};
          \draw (\i*\m,\i*\m) to (t\i.center);
        }
        \draw[<->] (0,-1) -- (\n,-1) node[midway,left=5pt, below]{ $|\sigma|m$ };
    \end{tikzpicture}
  \end{subfigure}
  \hfill
  \begin{subfigure}[b]{0.45\textwidth}
  \centering
    \begin{tikzpicture}[scale =0.2,rotate=90]
      \def\k{4}
      \def\m{3}
      \def\inner{5}
      \pgfmathsetmacro{\n}{\k*\m}
      \pgfmathsetmacro{\outer}{\n-1+\inner}
      \pgfmathsetmacro{\mminus}{\m-1}
      \pgfmathsetmacro{\kminus}{\k-1}
      \pgfmathsetmacro{\phi}{360/(\n)}
        \foreach \r in {\inner,...,\outer} \draw (0,0) circle (\r);
        \begin{pgfonlayer}{foreground}
            \foreach \i in {0,...,\n}{
               \draw (\phi*\i:\inner) -- (\phi*\i:\outer);
               \foreach \j in {\inner,...,\outer} \node[tinyNode] at (\phi*\i:\j) {};
               \node[tinygateNode] at (\phi*\i:\inner) {};
            }
        \end{pgfonlayer}
        \foreach \i in {0,...,\mminus}{
          \draw[polyEdge, draw=blue,-] (\phi*\i:\inner) -- (3*\m*\phi-\phi*\i-\phi:\inner);
          \draw[polyEdge,-] (\phi*\i+\m*\phi:\inner) -- (4*\m*\phi-\phi*\i-\phi:\inner);
        }
    \end{tikzpicture}  
  \end{subfigure}
  \caption{Two graphs containing $\sw(\PP)\setminus \set{r}$ as minor, where $\PP$ is the polygonal embedding from \cref{fig:PU} with $m=3$ and $\sigma = a_1a_2\bar{a_1}\bar{a_2}$. Both are minor-universal for all (small enough) graphs embeddable on surfaces having signature $\sigma$ (here the torus). The left one has about half as many vertices.}
  \label{fig:PU2}
\end{myfigure}

\begin{proofof}{\textbf{\cref{th:tec}}.}
We proceed in two steps, which are summarized and formalized by the next two technical lemmas, whose proofs are sketched hereafter and fully presented in \cref{sec:proofs}.

The first step (in \cref{lem:outer}) is to transform any polygonal embedding into a p-major polygonal embedding that is Hamiltonian and outerplane.
More precisely,

\begin{restatable}{lemma}{lemouter}\label{lem:outer}
  Every polygonal embedding $\Pi$ of size $(m,n)$ is a p-minor of a polygonal embedding of size $(m+2n,0)$.
\end{restatable}

The intuition is that all internal vertices of $\Pi$ can, in a sense, be pushed to the boundary of the outer face without significantly increasing the number of vertices on each side. It follows that we can consider that all the vertices belong to the outerface of the polygonal embedding, making it outerplanar.
This is done by first cutting along the edges of some spanning forest of the inner graph (rooted in a vertices of the outerface), creating new faces attached to the outerface.
We make it outerplanar by splitting each root of the forest, thereby opening each newly created face to the outerface.
During this process, each non-leaf vertex may appear multiple times (the number of newly created vertices is bounded by the number of edges in the forest, and hence by $n$).
To keep track of this vertex duplication, we add some edges that could be contracted later to get back to the original polygonal embedding. 
Other important transformations are then needed to embed those new edges (with special care of the orientation of each twin side)
and to obtain an outerplanar polygonal embedding of the right size and containing $\Pi$ as a p-minor.

The second step (in \cref{lem:final}) is to show that any Hamiltonian outerplane polygonal embedding is a p-minor of a given polygonal embedding $\PP$, which depending only on the size and signature of the input polygonal embedding. More precisely,

\begin{restatable}{lemma}{lemfinal}\label{lem:final}
  Every polygonal embedding $\Pi$ of size $(m,0)$ is a p-minor of $\PU_{\sigma(\Pi),m}$.
\end{restatable}

The proof of \cref{lem:final} is inspired by the embedding of Hamiltonian planar graphs with $n$ vertices into the $(n\times n)$-grid~\cite[Theorem~(1.3)]{RST94}.
From this embedding, one can easily show that outerplanar graphs with $n$ vertices are minors of a half-$(n\times n)$-grid. In our context, the diagonal of this half-grid is the place for the sides of our polygonal embedding, as depicted in \cref{fig:PU}. Interestingly, along the way, we give in \cref{prop:altRST} an alternative proof of a result of~\cite[Theorem~(1.4)]{RST94} that states that every planar graph with $n$ vertices is minor of a Hamiltonian planar graph with $2n$ vertices.

Overall, \cref{th:tec} is a straightforward combination of \cref{lem:outer} and \cref{lem:final}.
\end{proofof}

\paragraph{Estimating the size of the minor-universal graph.}

Putting it all together, we can give an accurate estimate of the number of vertices of $\sU_{n,\mS}$, better than the one provided by \cref{th:main}.
Given the system of loops $L$ with vertex $r$ as in \cref{lem:crossing_systemloops}, $G$ has an embedding $\Pi_G$ of  size $(m,n)$ with $m \le c \cdot |E(G)| \le c\cdot (3n+3g-6) < 3c\cdot (n+g)$, where $c = 4$ if $\mS$ is orientable and $c = 30$ if $\mS$ is non-orientable. 
From \cref{th:tec}, $\Pi_G$ is a p-minor of $\PU_{\sigma,m+2n} = \PU_{\sigma,(3c+2)n + O(g)}$, where $\sigma = \sigma(\Pi_G)$. 
From \cref{prop:PU}, the minor-universal graph $\sU_{n,\mS}$, which is $\sw(\PU_{\sigma,(3c+2)n + O(g)}) \setminus \set{r}$, has at most $\frac{1}{2} |\sigma|^2 (m+2n)^2 = 2g^2((3c+2)n + O(g))^2$ vertices since $\sigma$ is minimal and $|\sigma| = 2g$.

For the sake of simplicity, assume that $g = o(n)$. The number of vertices of $\sU_{n,\mS}$ is therefore (plugging $c=4$ or $c=30$):
%
\begin{itemize}[noitemsep]
  \item $(392+o(1)) \cdot (gn)^2$ if $\mS$ is orientable; and
  \item $(16928+o(1)) \cdot (gn)^2$ if $\mS$ is non-orientable.
\end{itemize}

Obviously, the constant $c = 30$ for the non-orientable case plays an important role in the final result. However, it can be improved as follows. According to Fuladi~\cite{Fuladi23p}, starting with the orienting curve of Matou{\v{s}}ek et al.~\cite{MSTW13}, the construction of Fuladi et al.~\cite{FHdM24} can also produce, in any non-orientable surface of Euler genus $g$, a system of loops providing a polygonal schema of (non-canonical) signature $a_0 a_0 a_1 a_2 \bar{a_1} \bar{a_2} \dots a_{g-2} a_{g-1} \bar{a}_{g-2} \bar{a}_{g-1}$. 
In other words, the polygonal schema consists in one cross-cap and $g-1$ handles. This is compatible with the statement of \cite[Theorem~6.3]{GX13}.
Furthermore, all the loops, but the first one corresponding to $a_0$, cut at most~$4$ times each edge of the embedded graph. For the side $a_0$, this is three times more, so~$12$ crossings at most per edge. 
Then, we can adapt the polygonal embedding $\PU_{\sigma,m+2n}$ so that the two first sides have $3c \cdot (n+g) = 36 (n+g)$ non-corner vertices (so with $c=12$), and the $|\sigma|-2$ remaining sides have $3c \cdot (n+g) = 12 (n+g)$ non-corner vertices (so with $c=4$). 
Again, assuming $g = o(n)$, the number of vertices is then at most
$2(g-1)^2 ((3\cdot 4 + 2)n + O(g))^2 + O(n+g)^2 = (392 + o(1))\cdot (gn)^2$, as much as the orientable case.


\section{Proofs of the lemmas}
\label{sec:proofs}

Throughout the proofs, if $H$ is a minor of $G$, it will sometimes be more convenient to say that $G$ is a \emph{major} of $H$. Similarly, we say that a polygonal embedding $\Pi_G$ is a \emph{p-major} of $\Pi_H$ if $\Pi_H$ is a p-minor of $\Pi_G$.

Since the plane graphs we will consider are actually polygonal embeddings, the boundary of the outerface is a cycle. Whenever we talk about the outerface of a plane graph, we refer to this cycle. 
We say that two plane graphs $H$ and $G$, where $H$ is a minor of $G$, have the \emph{same outerface} if $H$ is obtained from $G$ without removing any vertex or edge from the outerface of $G$, and without contracting any edge with both endpoints on the outerface of $G$. It follows that:

\begin{property}\label{property:p-minorOuterface}
  If a polygonal embedding $\Pi_H$ is a minor of polygonal embedding $\Pi_G$ with same outerface, border and signature, then $\Pi_H$ is a p-minor of $\Pi_G$.
\end{property}

Recall that a polygonal embedding is nothing else than a plane graph with a given cycle outerface, border and signature. \cref{property:p-minorOuterface} allows us to consider that we can work on a plane graph, as long as we keep the same outerface and border when constructing a major.

In \cref{subsec:planar}, we consider plane graphs with a fixed cycle outerface (i.e., with an outerface whose boundary is a cycle), and we give tools to construct a plane major that preserves this outerface. These tools will be used later in the construction of the p-major of $\Pi$ in \cref{subsec:proofLemOuter}.

\subsection{Considering plane graphs}
\label{subsec:planar}

Let $G$ be a planar graph $G$ with a given planar embedding and having a cycle outerface $O$. We denote by $n = |V(G)| - |V(O)|$ the number of internal vertices of $G$. 

Since the boundary of the outerface of $G$ is a connected graph (being a cycle), we can assume w.l.o.g. that $G$ is connected. If this is not the case, we can, for instance, triangulate all faces except the outer face. By doing this we get a plane major of $G$ without altering its number of vertices nor its cycle outerface.


Let $\cF$ be a forest, subgraph of $G$, composed of $k$ trees denoted by $T_1,\dots,T_k$, and constructed as follows:
\begin{itemize}[noitemsep]
  \item Start with a spanning tree $T$ of $G$;
  \item For any path in $T$ connecting two vertices of $O$, remove one of its edges;
  \item Remove all remaining isolated vertices that are in $O$; and
  \item Root each $T_i$ at the only one vertex of $T_i \cap O$, denoted
    by $r_i$.
\end{itemize}

Observe that $\cF$ spans all internal vertices of $G$, thus $|V(F)| = n + k$ and $|E(F)| = n$. Since we consider a planar major of $G$, it is quite tempting to add edges to force $\cF$ to be a single tree. 
However, we cannot do that and suppose that the inner vertices induced a connected subgraph of $G$, as there could be an edge between two vertices of the outerface separating the internal vertices into two parts.


Now we construct a major of $G$, denoted by $G_1$, that consists in \emph{blowing up} each tree $T_i$ into a graph $T'_i = C_i \cup E_i$, composed of a cycle (or an edge) $C_i$ plus a set $E_i$ of extra edges. More precisely, blowing up $T_i$ consists in traversing the tree according to a plane Euler tour from $r_i$ (or in other words, a walk along the boundary of the outerface of $T_i$ considered here as a single graph). 
The cycle $C_i$ is constructed iteratively by adding a new vertex at each vertex of $T_i$ visited along this tour, two consecutive vertices on this tour being connected by an edge of $C_i$. See \cref{fig:blowup} for an illustration. 
Since $\cF$ has no isolated vertices, $T_i$ has at least one edge. If $T_i$ has exactly one edge, then $C_i$ consists of one single edge. 
The set $E_i$ consists of all edges connecting any two vertices $u,v$ of $C_i$ that correspond to the same visited vertex of $T_i$ and appear consecutively during the traversal. In other words, we add in $E_i$ an edge between $u,v$ of $C_i$ if they corresponds to same visited vertex $w$ of $T_i$ and that none of the vertices between $u$ and $v$ in $C_i$ corresponds to $w$.
This completes the description of blowing up $T_i$ into $T'_i$. Note that $T'_i$ is outerplanar.

\begin{myfigure}
  \scalebox{1}{\begin{picture}(0,0)%
\includegraphics{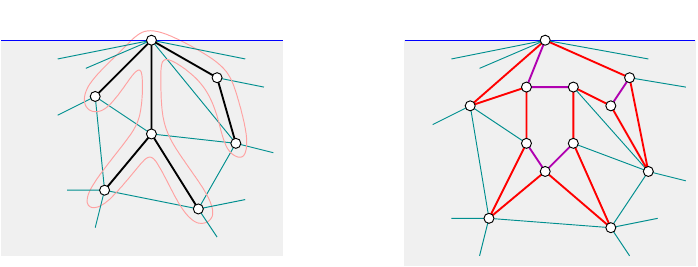}%
\end{picture}%
\setlength{\unitlength}{3947sp}%
\begingroup\makeatletter\ifx\SetFigFont\undefined%
\gdef\SetFigFont#1#2#3#4#5{%
  \reset@font\fontsize{#1}{#2pt}%
  \fontfamily{#3}\fontseries{#4}\fontshape{#5}%
  \selectfont}%
\fi\endgroup%
\begin{picture}(5574,2133)(589,-2323)
\put(4921,-1121){\makebox(0,0)[lb]{\smash{{\SetFigFont{11}{13.2}{\familydefault}{\mddefault}{\updefault}{\color[rgb]{.69,0,.69}$E_i$}%
}}}}
\put(4876,-361){\makebox(0,0)[lb]{\smash{{\SetFigFont{11}{13.2}{\familydefault}{\mddefault}{\updefault}{\color[rgb]{0,0,0}$r_i$}%
}}}}
\put(1726,-361){\makebox(0,0)[lb]{\smash{{\SetFigFont{11}{13.2}{\familydefault}{\mddefault}{\updefault}{\color[rgb]{0,0,0}$r_i$}%
}}}}
\put(4501,-1411){\makebox(0,0)[lb]{\smash{{\SetFigFont{11}{13.2}{\familydefault}{\mddefault}{\updefault}{\color[rgb]{1,0,0}$C_i$}%
}}}}
\put(751,-361){\makebox(0,0)[lb]{\smash{{\SetFigFont{11}{13.2}{\familydefault}{\mddefault}{\updefault}{\color[rgb]{0,0,0}$O$}%
}}}}
\put(751,-1486){\makebox(0,0)[lb]{\smash{{\SetFigFont{11}{13.2}{\familydefault}{\mddefault}{\updefault}{\color[rgb]{0,0,0}$G$}%
}}}}
\put(1726,-2161){\makebox(0,0)[lb]{\smash{{\SetFigFont{11}{13.2}{\familydefault}{\mddefault}{\updefault}{\color[rgb]{0,0,0}$T_i$}%
}}}}
\put(3976,-1561){\makebox(0,0)[lb]{\smash{{\SetFigFont{11}{13.2}{\familydefault}{\mddefault}{\updefault}{\color[rgb]{0,0,0}$G_1$}%
}}}}
\put(4901,-2211){\makebox(0,0)[lb]{\smash{{\SetFigFont{11}{13.2}{\familydefault}{\mddefault}{\updefault}{\color[rgb]{0,0,0}$T'_i$}%
}}}}
\end{picture}%
}
  \caption{Blowing up the tree $T_i$ (with black edges) into
    $T'_i = C_i \cup E_i$ (with red and violet edges). The root $r_i$ of $T_i$ is now the root of $C_i$ in $T'_i$.
  }
  \label{fig:blowup}
\end{myfigure}

The graph $G_1$ is then obtained from $G$ by replacing each $T_i$ with its blown-up version $T'_i$. This is possible because the walk along the boundary of the outerface of $T_i$ and $T'_i$ are isomorphic, so $T'_i$ can be plugged into $T_i$ by keeping the plane embedding of $G$.

\begin{claim}\label{claim:G1}
  The graph $G_1$ has the following properties:
  \begin{itemize}[noitemsep]
    \item $G_1$ is a plane major of $G$ with same outerface $O$;
    \item $C_1, \dots, C_k$ form a partition of internal vertices of $G_1$.
    \item Each $C_i$ contains exactly one vertex in $O$, called its \emph{root};
    \item $G_1$ has $2n - k$ internal vertices;
  \end{itemize}
\end{claim}

\begin{proof}
 It is straightforward to verify that $T'_i$ is a major of $T_i$, since contracting all edges in $E_i$ yields a graph isomorphic to $T_i$. By doing this for every $T'_i$ in $G_1$ we get exactly $G$ (up to some isomorphism) since these are the only differences.
  
 We have seen that the trees $T_i$ of $\cF$ form a partition of the internal vertices of $G$. It follows that $T'_i$ forms a partition of the internal vertices of $G_1$, and $C_i$ as well since $V(T'_i) = V(C_i)$.
  
 Each $T'_i$ intersects $O$ in exactly one vertex, $r_i$, which was the root of $T_i$. In particular, no edge of $T'_i$ belongs to $O$. So, contracting edges of $E_i$ to get $T_i$ cannot affect the outerface, and thus $G_1$ and $G$ have the same outerface $O$.

 To obtain $C_i$ from $T_i$, each edge of $T_i$ is traversed twice. It follows that $|V(C_i)| = 2|E(T_i)|$. We have seen that $|E(F)| = \sum_{i=1}^k |E(T_i)| = n$. Therefore, the number of internal vertices of $G_1$ is $\sum_{i=1}^k (|V(C_i)| - 1) = 2n - k$ since do not count roots.
\end{proof}

We now construct a major $G_2$ from $G_1$ by applying a \emph{root-splitting} operation at every root of $G_1$, as described below. We refer to \cref{fig:split} for an illustration. 
Consider a cycle $C_i$ with root $r_i$ in $G_1$, and let $u,v$ be the neighbors of $r_i$ such that both edges $u-r_i$ and $r_i-v$ belong to the boundary of the outerface.
We consider a cyclic ordering of the neighborhood of $r_i$ around $r_i$, starting at $u$ and ending at $v$. According to this ordering, let $w$ be the first neighbor of $r_i$ such that $r_i-w$ belongs to $C_i$, and let $z$ be the neighbor of $r_i$ in $T_i'$ that is right after $w$. If $C_i$ consists of one edge, we let $z = w$. Notice that the edge $r_i-z$ belongs to either $C_i$ or $E_i$.
The root-splitting operation of $r_i$ consists in: (1) replacing $r_i$ by an edge $a_i-b_i$, called the \emph{anchor} corresponding to $r_i$; (2) reconnecting to $a_i$ the neighbors of $r_i$ going from $u$ to $w$; and (3) reconnecting to $b_i$ the neighbors of $r_i$ going from $z$ to $v$. 

\begin{myfigure}
  \scalebox{1}{\begin{picture}(0,0)%
\includegraphics{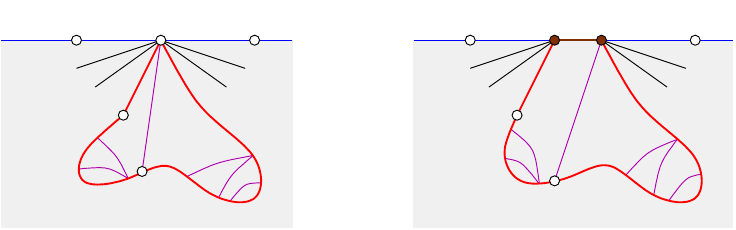}%
\end{picture}%
\setlength{\unitlength}{3947sp}%
\begingroup\makeatletter\ifx\SetFigFont\undefined%
\gdef\SetFigFont#1#2#3#4#5{%
  \reset@font\fontsize{#1}{#2pt}%
  \fontfamily{#3}\fontseries{#4}\fontshape{#5}%
  \selectfont}%
\fi\endgroup%
\begin{picture}(5874,1833)(514,-2023)
\put(4426,-1111){\makebox(0,0)[lb]{\smash{{\SetFigFont{11}{13.2}{\familydefault}{\mddefault}{\updefault}{\color[rgb]{0,0,0}$w$}%
}}}}
\put(1876,-1336){\makebox(0,0)[lb]{\smash{{\SetFigFont{11}{13.2}{\familydefault}{\mddefault}{\updefault}{\color[rgb]{.69,0,.69}$E_i$}%
}}}}
\put(5326,-1261){\makebox(0,0)[lb]{\smash{{\SetFigFont{11}{13.2}{\familydefault}{\mddefault}{\updefault}{\color[rgb]{.69,0,.69}$F_i$}%
}}}}
\put(2401,-1186){\makebox(0,0)[lb]{\smash{{\SetFigFont{11}{13.2}{\familydefault}{\mddefault}{\updefault}{\color[rgb]{1,0,0}$C_i$}%
}}}}
\put(601,-361){\makebox(0,0)[lb]{\smash{{\SetFigFont{11}{13.2}{\familydefault}{\mddefault}{\updefault}{\color[rgb]{0,0,0}$O$}%
}}}}
\put(601,-1261){\makebox(0,0)[lb]{\smash{{\SetFigFont{11}{13.2}{\familydefault}{\mddefault}{\updefault}{\color[rgb]{0,0,0}$G_1$}%
}}}}
\put(3976,-1411){\makebox(0,0)[lb]{\smash{{\SetFigFont{11}{13.2}{\familydefault}{\mddefault}{\updefault}{\color[rgb]{0,0,0}$G_2$}%
}}}}
\put(5926,-1186){\makebox(0,0)[lb]{\smash{{\SetFigFont{11}{13.2}{\familydefault}{\mddefault}{\updefault}{\color[rgb]{1,0,0}$P_i$}%
}}}}
\put(4876,-361){\makebox(0,0)[lb]{\smash{{\SetFigFont{11}{13.2}{\familydefault}{\mddefault}{\updefault}{\color[rgb]{.5,.17,0}$a_i$}%
}}}}
\put(4201,-361){\makebox(0,0)[lb]{\smash{{\SetFigFont{11}{13.2}{\familydefault}{\mddefault}{\updefault}{\color[rgb]{0,0,0}$u$}%
}}}}
\put(1051,-361){\makebox(0,0)[lb]{\smash{{\SetFigFont{11}{13.2}{\familydefault}{\mddefault}{\updefault}{\color[rgb]{0,0,0}$u$}%
}}}}
\put(1276,-1111){\makebox(0,0)[lb]{\smash{{\SetFigFont{11}{13.2}{\familydefault}{\mddefault}{\updefault}{\color[rgb]{0,0,0}$w$}%
}}}}
\put(5251,-361){\makebox(0,0)[lb]{\smash{{\SetFigFont{11}{13.2}{\familydefault}{\mddefault}{\updefault}{\color[rgb]{.5,.17,0}$b_i$}%
}}}}
\put(6001,-361){\makebox(0,0)[lb]{\smash{{\SetFigFont{11}{13.2}{\familydefault}{\mddefault}{\updefault}{\color[rgb]{0,0,0}$v$}%
}}}}
\put(1726,-361){\makebox(0,0)[lb]{\smash{{\SetFigFont{11}{13.2}{\familydefault}{\mddefault}{\updefault}{\color[rgb]{0,0,0}$r_i$}%
}}}}
\put(2476,-361){\makebox(0,0)[lb]{\smash{{\SetFigFont{11}{13.2}{\familydefault}{\mddefault}{\updefault}{\color[rgb]{0,0,0}$v$}%
}}}}
\put(1576,-1786){\makebox(0,0)[lb]{\smash{{\SetFigFont{11}{13.2}{\familydefault}{\mddefault}{\updefault}{\color[rgb]{0,0,0}$z$}%
}}}}
\put(4951,-1861){\makebox(0,0)[lb]{\smash{{\SetFigFont{11}{13.2}{\familydefault}{\mddefault}{\updefault}{\color[rgb]{0,0,0}$z$}%
}}}}
\end{picture}%
}
  \caption{Splitting the root $r_i$ of $C_i$ into the anchor $a_i-b_i$.}
  \label{fig:split}
\end{myfigure}

We denote by $P_i$ the path going from $a_i$ to $b_i$ with internal vertices in $V(C_i)\setminus\{r_i\}$, and by $F_i$ the new set of edges corresponding to $E_i$ in $G_2$. Note that $C_i$ and $P_i$ differ by $r_i,a_i,b_i$, and that $E_i$ and $F_i$ differ by the endpoints $a_i,b_i$ in $G_2$ that play the role of $r_i$ in $G_1$. In particular, if $C_i$ is an edge $r_i-w$, we get a triangle $a_i-w-b_i$.

\begin{claim}\label{claim:G2}
  The graph $G_2$ is a Hamiltonian plane major of $G_1$ with $|V(O)| + 2n$ vertices.
\end{claim}

\begin{proof}
  Clearly, $G_2$ is a planar major of $G_1$, as contracting every anchor of $G_2$ yields $G_1$.

  To construct a Hamiltonian cycle $C$ in $G_2$, we start with the cycle outerface of $G_2$, and replace each anchor $a_i-b_i$ by $P_i$. Since by \cref{claim:G1}, the $C_i$'s partition all internal vertices of $G_1$, and $V(C_i)\setminus\{r_i\}=V(P_i)\setminus\{a_i,b_i\}$ for each $i$, it holds that the internal vertices of $G_2$ are spanned by the $P_i$'s, which are pairwise disjoint. Thus $C$ forms a Hamiltonian cycle for $G_2$.

  The graph $G_1$ has $|V(O)|$ vertices on its outerface $O$ and $2n - k$ internal vertices (by \cref{claim:G1}). Splitting each of the $k$ roots adds $k$ vertices in $G_2$. Therefore, $|V(G_2)| = |V(G_1)| + k = (|V(O)| + (2n-k)) + k = |V(O)| + 2n$.
\end{proof}

Note that root-splitting operation alters the outerface. So, for the proof of \cref{lem:outer}, that needs to consider polygonal embeddings whose twin sides have to be of same length (among other things), we will need to make extra transformation on the outerface as explained later in \cref{subsec:proofLemOuter}.

Before that, we observe that the two previous transformations (by \cref{claim:G1} and \cref{claim:G2}) provide an alternative proof of~\cite[Theorem~(1.4)]{RST94} that states that every $n$-vertex planar graph is a minor of a planar Hamiltonian graph with at most
$2n$ vertices. Combined with~\cite[Theorem~(1.3)]{RST94}, that states that $n$-vertex planar Hamiltonian graphs are minor of an $(n\times n)$-grid, we can derive the minor-universal grid theorem of planar graphs~\cite[Theorem~(1.5)]{RST94}. 

In this context, a \emph{circuit} in a graph extends the notion of cycle to subgraph composed of one single vertex or edge.
A circuit is \emph{separating} if its deletion increases the number of connected components of the graph, and it is \emph{non-separating} otherwise.
It is not difficult to see that every graph has a non-separating circuit with at least one vertex.

\begin{theorem}\label{prop:altRST}
  Every planar graph with $n$ vertices and with a non-separating circuit of $k$ vertices is minor of a Hamiltonian planar graph with at most $2n-k$ vertices. In particular, every triangulation with $n\ge 4$ vertices is minor of a Hamiltonian planar graph with at most $2n - 4$ vertices.
\end{theorem}

\begin{proof}
  Let $G$ be a plane graph with $n$ vertices and with a non-separating circuit of $k$ vertices, denoted by $O$. W.l.o.g., we can assume that $O$ is a cycle, i.e., $k \ge 3$, since otherwise we can triangulate $G$ and consider its triangle outerface as non-separating circuit (recall that $G$ has at least four vertices).
  Moreover, the triangulation is a major of $G$ and we will get a minor of $2n - 3$ vertices that is less than $2n - k$ if $k < 3$.

  First, we embed $G$ in the plane such that all the vertices of $G\setminus O$ are inside $O$. 
  Note that $O$ is not necessarily the boundary of the outerface of this embedding, since there can be some chords outside of $O$. Denote by $C$ the set of these chords. Let $G'$ be the plane graph obtained from $G\setminus C$ such that $O$ is the boundary of its outerface which is a cycle.

  Since $G'$ is a plane graph with cycle outerface $O$, we can construct the major $G_1$ as in \cref{claim:G1}, and the Hamiltonian plane major $G_2$ from $G_1$ (by \cref{claim:G2}). Due to the anchor, the outerface of $G_2$ is a subdivision of $O$. 
  Therefore, we can add all the chords of $C$ in $G_2$ that we can embedded outside the outerface of $G_2$, while preserving its planarity. 
  Thus, $G_2 \cup C$ is the desired major of $G$.

  The graph $G'$ has $n' = n - k$ internal vertices by construction. So, by \cref{claim:G2}, $G_2$ (and $G_2 \cup C$ as well) has $k + 2n' = 2n - k$ vertices as claimed.

  If $G$ is a triangulation with at least four vertices, then $G$ has a non-separating circuit of length $4$. 
  Indeed, the outerface is a non-separating triangle that can be increased by one using any incident internal face. So, the above property construction applies with $k = 4$.
\end{proof}

Actually, the proof of~\cite[Theorem~(1.4)]{RST94} gives the same upper bound of $2n - 4$ vertices (assuming $n\ge 4$), but relies on Whitney's Theorem\footnote{Every planar triangulation without separating triangle has a Hamiltonian cycle.}~\cite{Whitney31}.
In contrast \cref{prop:altRST} gives an explicit and direct construction. 
We believe that it gives an interesting construction of the major where the Hamiltonian cycle cuts the graph into two sides. Inside the cycle are the copies of edges that where initially in $G$, and outside of the cycle are the edges introduced by the operations of blowing up trees and splitting roots: contracting those edges results in $G$.

\subsection{Proof of Lemma~\ref{lem:outer}}
\label{subsec:proofLemOuter}

The goal of this section is to show that:

\lemouter*

For every polygonal embedding, we define a total ordering $\prec$ on the vertices of its cycle outerface as follows: $u \prec v$ if $u$ is visited before $v$ when traversing the cycle outerface starting from the first corner and in the direction such that corners are encountered in the same order they appear in the border of the polygonal embedding.

Let $\Pi$ be a polygonal embedding of size $(m,n)$. Up to replacing $\Pi$ with a p-major, we may assume that each corner of $\Pi$ lies on the boundary of an internal triangle; if this is not the case, we can simply connect the two neighbors of each corner (which must have degree two) with an edge. This does not alter the size of $\Pi$.
This is to prevent any further edge insertion (for example while triangulating the graph) that could increase the degree of corners which must be exactly two in any polygonal embedding.

As the boundary of the outerface of $\Pi$ is a cycle, the construction in \cref{claim:G1} applies to $\Pi$ playing the role of $G$.
So, let $\Pi_1$ be the polygonal embedding obtained from \cref{claim:G1}, with $\Pi$ as the role of $G$ and $\Pi_1$ as the role of $G_1$. The embedding $\Pi_1$ admits $\Pi$ as a minor while preserving the cycle outerface (border and signature does not matter for \cref{claim:G1}). Thus $\Pi$ is a p-minor of $\Pi_1$. Keeping the same notation as in \cref{claim:G1}, $\Pi_1$ has exactly $2n - k$ internal vertices that are partitioned by the cycles $C_1,\dots,C_k$, each $C_i$ intersecting the outerface in one vertex, its root $r_i$. 

Unfortunately, as noted above, the transformation from $G_1$ to $G_2$ described in \cref{claim:G2} does not preserve the polygonal embedding structure: twin edges in $G_1$ cease to be twins in $G_2$ owing to the newly created anchors. Consequently, this breaks the p-major sequence under construction.
To preserve the polygonal embedding structure, we adapt the definition of root-splitting to this setting by modifying the twin sides accordingly whenever a root is split, thereby ensuring that every vertex on a given side has a twin on its corresponding twin side.
We refer to this procedure as a \emph{twin-splitting} operation, which is defined more precisely below. This operation is applied iteratively to each root, yielding a sequence of polygonal embeddings $\Pi_1^{0},\dots,\Pi_1^{k}$, where $\Pi_1^0 = \Pi_1$, such that each $\Pi_1^t$ is obtained from $\Pi_1^{t-1}$ via a twin-splitting operation.


More precisely, at step~$t$, consider three consecutive vertices $u-r-v$ on a side of $\Pi_1^{t-1}$, where $r$ is a root and $u\prec r \prec v$. Let $u',r',v'$ denote the twin vertices of $u,r,v$, respectively.
Since a root cannot be a corner and every side has a twin side, the vertices $u,v,u',v'$ are well defined.
NNote that if the twin sides have the same orientation in the signature (e.g., if they correspond to a handle on the surface under a canonical signature), then $v'\prec r'\prec u'$; whereas if they have opposite orientations (e.g., in a cross-cap), then $u'\prec r'\prec v'$.
The twin-splitting of $r$ consists in splitting the root $r$ followed by (see \cref{fig:twin}):
\begin{itemize}[noitemsep]
  \item a subdivision of the edge $r'-v'$, if $r \prec r'$ or $v'\prec r'$; or
  \item a subdivision of the edge $u'-r'$, if $r' \prec r$ and $r'\prec v'$.
\end{itemize}
In the former case, the edge $r'-v'$ is replaced by the path $r'-s-v'$, and $r'-s$ becomes the twin edge of the anchor $a-b$ corresponding to $r$. In the latter case, the edge $u'-r'$ is replaced by the path $u'-s-r'$, and $s-r'$ becomes the twin edge of $a-b$. All other twin relations remain unchanged.
See \cref{fig:twin} for an illustration. See \cref{fig:twin} for an illustration. Note that $r'$ may itself be a root that will eventually be split at a subsequent step $t' > t$.

\begin{myfigure}
  \scalebox{1}{\begin{picture}(0,0)%
\includegraphics{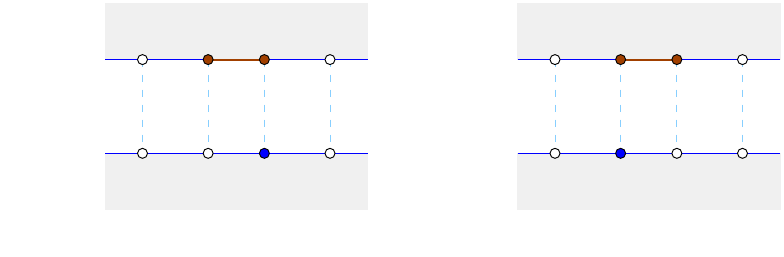}%
\end{picture}%
\setlength{\unitlength}{3947sp}%
\begingroup\makeatletter\ifx\SetFigFont\undefined%
\gdef\SetFigFont#1#2#3#4#5{%
  \reset@font\fontsize{#1}{#2pt}%
  \fontfamily{#3}\fontseries{#4}\fontshape{#5}%
  \selectfont}%
\fi\endgroup%
\begin{picture}(6252,2035)(-89,-2234)
\put(4276,-2161){\makebox(0,0)[lb]{\smash{{\SetFigFont{11}{13.2}{\familydefault}{\mddefault}{\updefault}{\color[rgb]{0,0,0}case 2: $r'\prec r$ and $r'\prec v'$}%
}}}}
\put(5776,-1711){\makebox(0,0)[lb]{\smash{{\SetFigFont{11}{13.2}{\familydefault}{\mddefault}{\updefault}{\color[rgb]{0,0,0}$v'$}%
}}}}
\put(4276,-511){\makebox(0,0)[lb]{\smash{{\SetFigFont{11}{13.2}{\familydefault}{\mddefault}{\updefault}{\color[rgb]{0,0,0}$u$}%
}}}}
\put(5776,-511){\makebox(0,0)[lb]{\smash{{\SetFigFont{11}{13.2}{\familydefault}{\mddefault}{\updefault}{\color[rgb]{0,0,0}$v$}%
}}}}
\put(5251,-1711){\makebox(0,0)[lb]{\smash{{\SetFigFont{11}{13.2}{\familydefault}{\mddefault}{\updefault}{\color[rgb]{0,0,0}$r'$}%
}}}}
\put(4801,-886){\makebox(0,0)[lb]{\smash{{\SetFigFont{11}{13.2}{\familydefault}{\mddefault}{\updefault}{\color[rgb]{.63,.25,0}$a$}%
}}}}
\put(5251,-886){\makebox(0,0)[lb]{\smash{{\SetFigFont{11}{13.2}{\familydefault}{\mddefault}{\updefault}{\color[rgb]{.63,.25,0}$b$}%
}}}}
\put(5026,-511){\makebox(0,0)[lb]{\smash{{\SetFigFont{11}{13.2}{\familydefault}{\mddefault}{\updefault}{\color[rgb]{0,0,0}$r$}%
}}}}
\put(1276,-511){\makebox(0,0)[lb]{\smash{{\SetFigFont{11}{13.2}{\familydefault}{\mddefault}{\updefault}{\color[rgb]{0,0,0}$\prec$}%
}}}}
\put(4576,-511){\makebox(0,0)[lb]{\smash{{\SetFigFont{11}{13.2}{\familydefault}{\mddefault}{\updefault}{\color[rgb]{0,0,0}$\prec$}%
}}}}
\put(5401,-511){\makebox(0,0)[lb]{\smash{{\SetFigFont{11}{13.2}{\familydefault}{\mddefault}{\updefault}{\color[rgb]{0,0,0}$\prec$}%
}}}}
\put(976,-2161){\makebox(0,0)[lb]{\smash{{\SetFigFont{11}{13.2}{\familydefault}{\mddefault}{\updefault}{\color[rgb]{0,0,0}case 1: $r\prec r'$ or $v'\prec r'$}%
}}}}
\put(5476,-1711){\makebox(0,0)[lb]{\smash{{\SetFigFont{11}{13.2}{\familydefault}{\mddefault}{\updefault}{\color[rgb]{0,0,0}$\prec$}%
}}}}
\put(4801,-1711){\makebox(0,0)[lb]{\smash{{\SetFigFont{11}{13.2}{\familydefault}{\mddefault}{\updefault}{\color[rgb]{0,0,0}$\prec$}%
}}}}
\put(1976,-1311){\makebox(0,0)[lb]{\smash{{\SetFigFont{11}{13.2}{\familydefault}{\mddefault}{\updefault}{\color[rgb]{0,0,1}$s$}%
}}}}
\put(4826,-1311){\makebox(0,0)[lb]{\smash{{\SetFigFont{11}{13.2}{\familydefault}{\mddefault}{\updefault}{\color[rgb]{0,0,1}$s$}%
}}}}
\put(-74,-736){\makebox(0,0)[lb]{\smash{{\SetFigFont{11}{13.2}{\familydefault}{\mddefault}{\updefault}{\color[rgb]{0,0,0}side of $r$}%
}}}}
\put(-74,-1411){\makebox(0,0)[lb]{\smash{{\SetFigFont{11}{13.2}{\familydefault}{\mddefault}{\updefault}{\color[rgb]{0,0,0}twin side}%
}}}}
\put(976,-1711){\makebox(0,0)[lb]{\smash{{\SetFigFont{11}{13.2}{\familydefault}{\mddefault}{\updefault}{\color[rgb]{0,0,0}$u'$}%
}}}}
\put(2476,-1711){\makebox(0,0)[lb]{\smash{{\SetFigFont{11}{13.2}{\familydefault}{\mddefault}{\updefault}{\color[rgb]{0,0,0}$v'$}%
}}}}
\put(976,-511){\makebox(0,0)[lb]{\smash{{\SetFigFont{11}{13.2}{\familydefault}{\mddefault}{\updefault}{\color[rgb]{0,0,0}$u$}%
}}}}
\put(2476,-511){\makebox(0,0)[lb]{\smash{{\SetFigFont{11}{13.2}{\familydefault}{\mddefault}{\updefault}{\color[rgb]{0,0,0}$v$}%
}}}}
\put(1501,-1711){\makebox(0,0)[lb]{\smash{{\SetFigFont{11}{13.2}{\familydefault}{\mddefault}{\updefault}{\color[rgb]{0,0,0}$r'$}%
}}}}
\put(1501,-886){\makebox(0,0)[lb]{\smash{{\SetFigFont{11}{13.2}{\familydefault}{\mddefault}{\updefault}{\color[rgb]{.63,.25,0}$a$}%
}}}}
\put(1951,-886){\makebox(0,0)[lb]{\smash{{\SetFigFont{11}{13.2}{\familydefault}{\mddefault}{\updefault}{\color[rgb]{.63,.25,0}$b$}%
}}}}
\put(1726,-511){\makebox(0,0)[lb]{\smash{{\SetFigFont{11}{13.2}{\familydefault}{\mddefault}{\updefault}{\color[rgb]{0,0,0}$r$}%
}}}}
\put(2101,-511){\makebox(0,0)[lb]{\smash{{\SetFigFont{11}{13.2}{\familydefault}{\mddefault}{\updefault}{\color[rgb]{0,0,0}$\prec$}%
}}}}
\put(4276,-1711){\makebox(0,0)[lb]{\smash{{\SetFigFont{11}{13.2}{\familydefault}{\mddefault}{\updefault}{\color[rgb]{0,0,0}$u'$}%
}}}}
\end{picture}%
}
  \caption{The twin-splitting operation for a root $r$: it combines splitting the root $r$ (in brown) and subdividing an edge (in blue) incident to its twin $r'$. Twin relations are represented with dotted lines. Note that in this representation, the focus on the twin sides creates the illusion that the inside part of the embedding (in gray) is a disconnected region, which is not true. However, in this representation choice, the inside part, that is locally planar, may appear twisted if fully represented in the case where the twin sides form a cross-cap in the surface.}
  \label{fig:twin}
\end{myfigure}

After completing this twin-splitting process, we denote by $\Pi_2$ (i.e., $\Pi_1^k$) the final embedding obtained from $\Pi_1$ after applying all twin-splitting operations in sequence.

\begin{claim}\label{claim:splitAnchor}
  The embedding $\Pi_2$ has the following properties:
  \begin{itemize}[noitemsep]
    \item $\Pi_2$ is a polygonal embedding that is p-major of $\Pi_1$;
    \item $\Pi_2$ has size $(m+k,2n-k)$;
    \item $P_1, \dots, P_k$ form a partition of the internal vertices of $\Pi_2$;
    \item During the twin-splitting process, no anchor is subdivided and no two anchors are twins.
  \end{itemize}
\end{claim}

\begin{proof}
  Each twin-splitting operation inserts edges and splits roots on the outerface of $\Pi_1$. Contracting these newly added edges yields $\Pi_1$, which shows that $\Pi_2$ is a major of $\Pi_1$.
  Furthermore, by construction of the twin-splitting procedure, every vertex and edge of a side in $\Pi_2$ has a well-defined twin. Hence, $\Pi_2$ is a polygonal embedding that is a p-major of $\Pi_1$.
 
  At each twin-splitting operation, the number of vertices of a side increases by at most~1 while the number of internal vertices does not change. 
  Since $\Pi_1$ contains $k$ roots, after applying the twin process, each side of $\Pi_2$ contains at most $m+k$ vertices (corners excluded) and $2n-k$ internal vertices. So, the size of $\Pi_2$ is $(m+k,2n-k)$.
  
  The internal vertices of $\Pi_1$ and $\Pi_2$ coincide. Moreover, after each root-splitting operation, the paths $P_i$ span the same internal vertices as the cycles $C_i$. Thus, by \cref{claim:G1}, the paths $P_i$ also partition the internal vertices of $\Pi_2$.

  It remains to prove the final assertion, namely that no two anchors become twins or are subdivided during the twin-splitting process. Suppose, for the sake of contradiction, that there exists a step $t'$ and a root $r'$ in $\Pi_1^{t'-1}$ such that twin-splitting $r'$ creates twin anchors or subdivides an anchor. Observe that the twin-splitting operation affects only the edges around $r'$ and its twin, thus this can happen only if the twin of $r'$ belongs to an anchor. This occurs only if there is some $t < t'$ such that $r'$ and its twin $r$ were both roots in $\Pi_1^{t-1}$ and we applied a twin-splitting operation on $r$ at step $t$.

  W.l.o.g. assume that $t$ is the step where $r$ is split into $a-b$, and let $s$ be the neighbor of $r'$ in $\Pi_1^t$ resulting from the twin-splitting of $r$ in $\Pi_1^{t-1}$.
  We use for $\Pi_1^{t-1}$ and $\Pi_1^{t}$ the notation as above and as in \cref{fig:twin}: in $\Pi_1^{t-1}$, we have $u-r-v$ on one side, with $u\prec r\prec v$ and with respective twins $u',r',v'$ ($r$ and $r'$ are both roots).
  In $\Pi_1^{t}$, we have $u-a-b-v$ with $u\prec a\prec b\prec v$, and $s$ adjacent to $r'$, resulting from the subdivision of an edge incident to $r'$ (of either $r'-v'$ or $u'-r'$). 

  We remark first that, while an anchor can be incident to the twin of an anchor or the twin of a root, it cannot be incident to another anchor nor a root.
  In particular, in $\Pi_1^{t}$, $r'$ is not incident to an anchor (since $r'$ is a root), and the only twin of an anchor incident to $r'$ is $r'-s$ (since the other edge of the border incident to $r'$ is twin to an edge incident to the anchor $a-b$).
  We claim that this property persists in $\Pi_1^{t'-1}$. Let $w'$ denote the vertex among $u',v'$ adjacent to $r'$, and let $w$ be its twin.
  Since $s$ is not a root nor the twin of a root, before step $t'$, the edge $r'-s$ is not transformed. Suppose that there is some step $t<t''<t'$ in which $w'-r'$ gets subdivided, say into $w'-s'-r'$. This can only happen if $w'$ or $w$ is a root. In both cases, the (twin of an) anchor is $w'-s'$, so the edge $s'-r'$ is still neither an anchor nor the twin of an anchor.
  
  Since we are only interested in the edges around $r'$ and its twin, and thanks to the previous observation and slight abuse of notation for $u',v'$, we can keep for $\Pi_1^{t'-1}$ the same notations as in $\Pi_1^{t}$, namely:  $u-a-b-v$ with $u\prec a\prec b\prec v$ on one side, and either $u'-r'-s-v'$ or $u'-s-r'-v'$ as their twin the other side.
  
  Let us now apply step $t'$ on $r'$ which produces the anchor $a'-b'$ in $\Pi_1^{t'}$.

  Assume first that $r'\prec r$ and $r'\prec v'$ in $\Pi_1^{t-1}$, i.e., the case 2 holds for the twin-splitting of $r$. In that case, we have the path $u'-s-r'-v'$ in $\Pi_1^{t'-1}$, and also $u\prec a\prec b\prec v$, where $a$ is the twin of $s$ and $b$ the twin of $r'$. 
  For the twin-split of $r'$, we are in the case 1, with $s,r',v',a,b,v$ taking the roles of $u,r,v,u',r',v'$ respectively. Indeed, the condition $r\prec r'$ rewrites in $r'\prec b$ which is true.
  Therefore, the twin-splitting of $r'$ subdivides the edge $b-v$ which cannot be an anchor by the previous remark. 
  The anchor $a'-b'$ corresponding to $r'$ is twin with a non-anchor as well. Thus, step $t'$ does not twin or subdivide any anchor in this case.
  
  Assume now that $r\prec r'$ or $v'\prec r'$, i.e., the case 1 holds for the twin-splitting of $r$ in $\Pi_1^{t-1}$. In that case, we have the path $u'-r'-s-v'$ in $\Pi_1^{t'-1}$, and also that $u\prec a\prec b\prec v$, where $a$ is the twin of $r'$ and $b$ the twin of $s$.
  
  If $v'\prec r'$, then we have $v'\prec s\prec r'\prec u'$.
  Then for the twin-splitting of $r'$, the case 1 holds with $s,r',u',b,a,u$ in the roles of $u,r,v,u',r',v'$ respectively. This is because the condition $v'\prec r'$ rewrites in  $u\prec a$ which is true.
  
  If $r\prec r'\prec v'$, then we have $a\prec r'$. 
  Then for the twin-splitting of $r'$, the case 2 holds with $u,r',s,u,a,b$ in the roles of $u,r,v,u',r',v'$ respectively. This is because the condition $r'\prec r$ and $r'\prec v'$ rewrites in  $a\prec r'$ and $a\prec b$ which is true.
  
  Thus, in both cases, the twin-splitting of $r'$ subdivides the edge $u-a$, which cannot be an anchor by the previous remark. Similarly, the anchor $a'-b'$ is paired with a non-anchor. Therefore, in all cases, step $t'$ neither creates twin anchors nor subdivides any anchor, a contradiction.
\end{proof}


The last property of \cref{claim:splitAnchor} (no anchor is subdivided) ensures that the roots in $\Pi_1$ correspond to anchors in $\Pi_2$.

To transform $\Pi_2$ into an outerplanar embedding, we apply a \emph{swapping} operation (defined below) to each of its anchors, thereby generating a sequence of polygonal embeddings $\Pi_2^{0},\dots,\Pi_2^{k}$ with $\Pi_2^0 = \Pi_2$, such that each $\Pi_2^t$ is constructed from $\Pi_2^{t-1}$ via one such operation.

The swapping operation is defined as follows. (We reuse the same notation as in \cref{fig:split} describing the root-splitting operation which has been used to make $\Pi_2$.) 
Let $a_i-b_i$ be an anchor of $\Pi_2^{t-1}$, $P_i$ be the path connecting $a_i$ to $b_i$, and $F_i$ be the edges not in $P_i$ and embedded inside the cycle $P_i \cup \set{a_i - b_i}$. Finally, let $a'_i,b'_i$ be the twin vertices of respectively $a_i,b_i$. 

To obtain $\Pi_2^t$, we first delete the edge $a_i-b_i$ and all edges of $F_i$. This way, the side containing $a_i-b_i$ is extended by $P_i$. Then, on the outerface of $\Pi_2^{t-1}$, we connect $a'_i$ to $b'_i$ by a copy $P'_i$ of $P_i$ and including a copy $F'_i$ of all the edges of $F_i$ such that the order of $P_i$ is preserved. 
It means that if $u$ is traversed before $v$ when going from $a_i$ to $b_i$ on $P_i$, then $u'$ is traversed before $v'$ when going from $a'_i$ in $b'_i$ in $P'_i$ where $u',v'$ are the copies in $P'_i$ of $u,v$ respectively. Furthermore, the twin of any vertex $u$ in $P_i$ is its copy of $u'$ in $P'_i$. This way, the side containing $a'_i-b'_i$ is extended by $P'_i$. See \cref{fig:swap} for an illustration.

\begin{myfigure}
  \scalebox{1}{\begin{picture}(0,0)%
\includegraphics{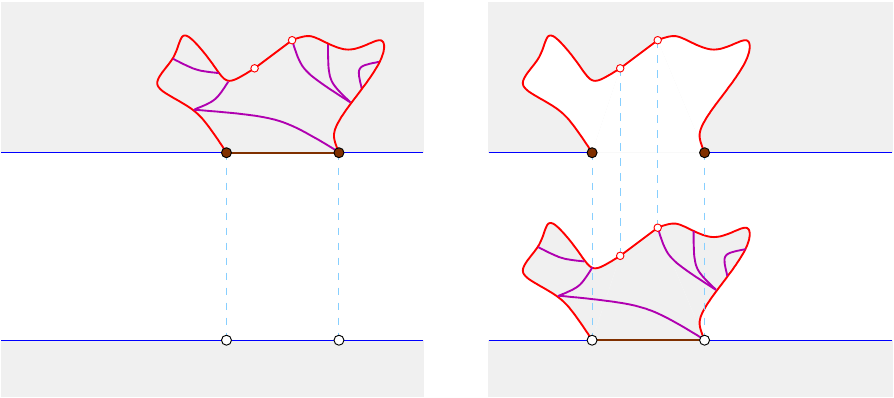}%
\end{picture}%
\setlength{\unitlength}{3947sp}%
\begingroup\makeatletter\ifx\SetFigFont\undefined%
\gdef\SetFigFont#1#2#3#4#5{%
  \reset@font\fontsize{#1}{#2pt}%
  \fontfamily{#3}\fontseries{#4}\fontshape{#5}%
  \selectfont}%
\fi\endgroup%
\begin{picture}(7149,3174)(2839,-2473)
\put(3001,-361){\makebox(0,0)[lb]{\smash{{\SetFigFont{11}{13.2}{\familydefault}{\mddefault}{\updefault}{\color[rgb]{0,0,0}inside $\Pi_2^{t-1}$}%
}}}}
\put(7501,-766){\makebox(0,0)[lb]{\smash{{\SetFigFont{11}{13.2}{\familydefault}{\mddefault}{\updefault}{\color[rgb]{.5,.17,0}$a_i$}%
}}}}
\put(8401,-766){\makebox(0,0)[lb]{\smash{{\SetFigFont{11}{13.2}{\familydefault}{\mddefault}{\updefault}{\color[rgb]{.5,.17,0}$b_i$}%
}}}}
\put(7741,254){\makebox(0,0)[lb]{\smash{{\SetFigFont{10}{12.0}{\familydefault}{\mddefault}{\updefault}{\color[rgb]{1,0,0}$u$}%
}}}}
\put(8063,479){\makebox(0,0)[lb]{\smash{{\SetFigFont{10}{12.0}{\familydefault}{\mddefault}{\updefault}{\color[rgb]{1,0,0}$v$}%
}}}}
\put(7051,-391){\makebox(0,0)[lb]{\smash{{\SetFigFont{11}{13.2}{\familydefault}{\mddefault}{\updefault}{\color[rgb]{1,0,0}$P_i$}%
}}}}
\put(9001,-2311){\makebox(0,0)[lb]{\smash{{\SetFigFont{11}{13.2}{\familydefault}{\mddefault}{\updefault}{\color[rgb]{0,0,0}inside $\Pi_2^{t}$}%
}}}}
\put(7501,-2296){\makebox(0,0)[lb]{\smash{{\SetFigFont{11}{13.2}{\familydefault}{\mddefault}{\updefault}{\color[rgb]{0,0,0}$a'_i$}%
}}}}
\put(8401,-2296){\makebox(0,0)[lb]{\smash{{\SetFigFont{11}{13.2}{\familydefault}{\mddefault}{\updefault}{\color[rgb]{0,0,0}$b'_i$}%
}}}}
\put(3001,-1336){\makebox(0,0)[lb]{\smash{{\SetFigFont{11}{13.2}{\familydefault}{\mddefault}{\updefault}{\color[rgb]{0,0,0}outerface of $\Pi_2^{t-1}$}%
}}}}
\put(8701,-961){\makebox(0,0)[lb]{\smash{{\SetFigFont{11}{13.2}{\familydefault}{\mddefault}{\updefault}{\color[rgb]{0,0,0}outerface of $\Pi_2^{t}$}%
}}}}
\put(7741,-1246){\makebox(0,0)[lb]{\smash{{\SetFigFont{10}{12.0}{\familydefault}{\mddefault}{\updefault}{\color[rgb]{1,0,0}$u'$}%
}}}}
\put(8063,-1021){\makebox(0,0)[lb]{\smash{{\SetFigFont{10}{12.0}{\familydefault}{\mddefault}{\updefault}{\color[rgb]{1,0,0}$v'$}%
}}}}
\put(7051,-1891){\makebox(0,0)[lb]{\smash{{\SetFigFont{11}{13.2}{\familydefault}{\mddefault}{\updefault}{\color[rgb]{1,0,0}$P'_i$}%
}}}}
\put(7951,-1591){\makebox(0,0)[lb]{\smash{{\SetFigFont{11}{13.2}{\familydefault}{\mddefault}{\updefault}{\color[rgb]{.69,0,.69}$F'_i$}%
}}}}
\put(4576,-1861){\makebox(0,0)[lb]{\smash{{\SetFigFont{11}{13.2}{\familydefault}{\mddefault}{\updefault}{\color[rgb]{0,0,0}$a'_i$}%
}}}}
\put(5476,-1861){\makebox(0,0)[lb]{\smash{{\SetFigFont{11}{13.2}{\familydefault}{\mddefault}{\updefault}{\color[rgb]{0,0,0}$b'_i$}%
}}}}
\put(3001,-2311){\makebox(0,0)[lb]{\smash{{\SetFigFont{11}{13.2}{\familydefault}{\mddefault}{\updefault}{\color[rgb]{0,0,0}inside $\Pi_2^{t-1}$}%
}}}}
\put(4576,-766){\makebox(0,0)[lb]{\smash{{\SetFigFont{11}{13.2}{\familydefault}{\mddefault}{\updefault}{\color[rgb]{.5,.17,0}$a_i$}%
}}}}
\put(5476,-766){\makebox(0,0)[lb]{\smash{{\SetFigFont{11}{13.2}{\familydefault}{\mddefault}{\updefault}{\color[rgb]{.5,.17,0}$b_i$}%
}}}}
\put(4126,-391){\makebox(0,0)[lb]{\smash{{\SetFigFont{11}{13.2}{\familydefault}{\mddefault}{\updefault}{\color[rgb]{1,0,0}$P_i$}%
}}}}
\put(4816,254){\makebox(0,0)[lb]{\smash{{\SetFigFont{10}{12.0}{\familydefault}{\mddefault}{\updefault}{\color[rgb]{1,0,0}$u$}%
}}}}
\put(5138,479){\makebox(0,0)[lb]{\smash{{\SetFigFont{10}{12.0}{\familydefault}{\mddefault}{\updefault}{\color[rgb]{1,0,0}$v$}%
}}}}
\put(5026,-91){\makebox(0,0)[lb]{\smash{{\SetFigFont{11}{13.2}{\familydefault}{\mddefault}{\updefault}{\color[rgb]{.69,0,.69}$F_i$}%
}}}}
\put(9001,-361){\makebox(0,0)[lb]{\smash{{\SetFigFont{11}{13.2}{\familydefault}{\mddefault}{\updefault}{\color[rgb]{0,0,0}inside $\Pi_2^{t}$}%
}}}}
\end{picture}%
}
  \caption{A swapping operation for the anchor $a_i-b_i$ in
    $\Pi_2^{t-1}$ (on the left) leading to a new embedding $\Pi_2^t$
    (on the right). Vertex $u$ is traversed before $v$ when going from
    $a_i$ to $b_i$ on $P_i$, and so for their twin vertices $u',v'$
    when going from $a'_i$ to $b'_i$ on $P'_i$. The vertices of $P_i$
    that were internal in $\Pi_2^{t-1}$ belong to the outerface of
    $\Pi_2^{t}$. Again, in this representation choice, the inside part (in gray), that is locally planar, may appear twisted if fully represented in the case where twin sides form a cross-cap in the surface.}
  \label{fig:swap}
\end{myfigure}

The last property of \cref{claim:splitAnchor} (no anchors are
twins) ensures that $a'_i-b'_i$ was not an anchor in
$\Pi_2^{t-1}$. So the swapping operation on the anchor $a_i - b_i$ leaves unchanged all the other anchors.

\def\iso{\cong}%

We denote by $\Pi_3$ the final embedding obtained from $\Pi_2$ by applying successively a swapping operation on each of its anchors. We write $A \iso B$ if $A$ and $B$ are homeomorphic embeddings on $\mS$.

\begin{claim}\label{claim:swapAnchor}
  The embedding $\Pi_3$ is a polygonal embedding p-major of $\Pi_2$ with $\sw(\Pi_3) \iso \sw(\Pi_2)$ and size $(m+2n,0)$.
\end{claim}

\begin{proof}
  Up to a permutation of the indices, we can assume that the swapping operation for $a_i-b_i$ is performed at step $i$ from $\Pi_2^{i-1}$. 
  It is easy to see that each operation does not create edge crossings, that the outerface remains a cycle, and that each vertex (edge) of each side has a well-defined twin vertex (edge) in its twin side. Therefore, $\Pi_2^i$ is a polygonal embedding, and also $\Pi_3$, which is equal to $\Pi_2^k$, by transitivity.
  
  To show that $\Pi_3$ has the claimed size, let $n_i = \card{V(P_i)\setminus\set{a_i,b_i}}$ denote the number of internal vertices of $P_i$ in $\Pi_2^{i-1}$. Note that this number is the same as the number of vertices of $P_i$ in $\Pi_2^0 = \Pi_2$ because $P_i$ is not altered by any of the previous steps $j<i$. Therefore, from \cref{claim:splitAnchor}, $\sum_{j=1}^k n_j = 2n-k$.
  
  The swapping operation for $a_i-b_i$ increases by $n_i$ the number of vertices of its side and its twin side, and, at the same time, decreases by $n_i$ the number of internal vertices of $\Pi_2^{i-1}$. By \cref{claim:splitAnchor}, $\Pi_2$ has size $(m+k,2n-k)$.
  So after step $i$, the size of $\Pi_2^{i}$ is $(m+k + \sum_{j=1}^i n_j, n-2k - \sum_{j=1}^i n_j)$. 
  It follows that $\Pi_3$, which is $\Pi_2^k$, has size $(m+2n,0)$ since $\sum_{j=1}^k n_j = 2n-k$.
  
  It is clear that signature and border, i.e., the sequence of corners, are not altered by any swapping operation. 
  SIt remains to show that $\sw(\Pi_3) \iso \sw(\Pi_2)$ (which will also establish that $\Pi_3$ is a p-major of $\Pi_2$).
  By transitivity, it suffices to show that $\sw(\Pi_2^{i}) \iso \sw(\Pi_2^{i-1})$.  Let $R_i$ be the plane graph composed of $P_i\cup\set{a_i-b_i}\cup F_i$ in $\Pi_2^{i-1}$ such that the boundary of its outerface is the cycle $P_i \cup\set{a_i-b_i}$. 
  Similarly, let $R'_i$ be the plane graph composed of $P'_i\cup\set{a'_i-b'_i}\cup F'_i$ in $\Pi_2^{i}$ with cycle outerface $P'_i \cup\set{a'_i-b'_i}$. Note that $V(R_i) = V(P_i)$ and $V(R'_i) = V(P'_i)$. For convenience, denote by $G^j = \sw(\Pi_2^{j})$ for $j\in\set{i-1,i}$. We want to show that $G^{i-1} \iso G^{i}$.
  
  Observe that in $G^{i-1}$ and in $G^{i}$, whenever the twin sides are merged, $a_i = a'_i$, $b_i = b'_i$ and thus the edge $a_i-b_i = a'_i-b'_i$ exists in both graphs (whereas $a_i-b_i$ exists only in $\Pi_2^{i-1}$). 
  Also, the way the vertices of $P_i$ are twins with the vertices $P'_i$ (by preserving the order when going from $a_i$ to $b_i$ in $P_i$), ensures that $P'_i = P_i$ and $F'_i = F_i$ in $G^i$. 
  It follows that $R_i \iso R'_i$, and since the swapping operation alters only $R_i$, $V(G^{i-1}) = V(G^i)$. Thus, if there is a difference between $G^{i-1}$ and $G^i$, it must be an edge.
  
  From previous equalities, if an edge $x-y$ lies in $R_i$ then the edge exists in $R'_i$ and thus in $G^i$. Conversely, if $x-y$ lies in $R'_i$, the edge exists in $R_i$ and thus in $G^{i-1}$. 
  If $x-y$ lies outside $R_i$ ($x,y$ may be both in $V(R_i)$ with an embedding not inside $R_i$), then $x-y$ lies also in $\Pi_2^i$ since the swapping operation alters only $R_i$ (and its inside).
  So $x-y \in E(G^{i-1})$ implies $x-y \in E(G^i)$ (even if $x-y$ belongs to some side of $\Pi_2^{i-1}$). 
  Conversely, if $x-y$ lies outside $R'_i$ in $\Pi_2^i$, then $x-y$ lies also in $\Pi_2^{i-1}$ (and outside of $R_i$). It follows that the edge sets of $G^{i-1}$ and $G^i$ are the same, proving that $G^{i-1} \iso G^i$.
\end{proof}

By transitivity of the p-minor relation, every polygonal embedding $\Pi$ of size $(m,n)$ is a p-minor of $\Pi_2$ (by \cref{claim:splitAnchor}) that is a p-minor of $\Pi_3$ of size $(m+2n,0)$ (by \cref{claim:swapAnchor}), which completes the proof of \cref{lem:outer}.

\subsection{Proof of \cref{lem:final}}
\label{subsec:proofLemFinal}

The goal of this section is to show that:

\lemfinal*

The proof is an adaptation of the proof of Robertson, Seymour and Thomas~\cite[Theorem~(1.3)]{RST94} to the case of outerplanar polygonal embeddings. While the original proof relies on grid majors, our setting requires the use of $\PP$.

We first recall the construction of the polygonal embedding $\PP$, which is defined for any signature $\sigma$ and integer $m\in\mathbb{N}$, and introduce the notation for its vertices.
To construct $\PP$, we start from a square half-grid of dimension $|\sigma|m$, that is with $|\sigma|m$ rows and $|\sigma|m$ columns. 
The vertices are column-row pairs of integers $(i,j)$ where $0\le j\le i < |\sigma|m$, and $(i,j)$ and $(i',j')$ are adjacent if and only if $|i-i'| + |j-j'| = 1$. 
Then, the vertices of the diagonal are connected by a cycle, $(i,i)$ being connected to $(i-1,i-1)$ modulo $|\sigma|m$. This cycle forms the cycle outerface of $\PP$. 
Finally, every $m$ edges along this cycle is subdivided by one vertex, a corner. More precisely, for each $i \in\range{0}{|\sigma|-1}$, we add a vertex $c_i$ between $(im,im)$ and $(im-1,im-1)$ modulo $|\sigma|m$. 
The border of $\PP$ is the sequence $(c_0,\dots,c_{|\sigma|-1})$, and the $i$th side of $\PP$, associated with the $i$th symbol of $\sigma$, is the path between $c_i$ and $c_{i+1 \bmod |\sigma|}$ on the cycle outerface. See \cref{fig:pu} and \cref{fig:PU} for illustrations.

\begin{myfigure}
  \scalebox{1}{\begin{picture}(0,0)%
\includegraphics{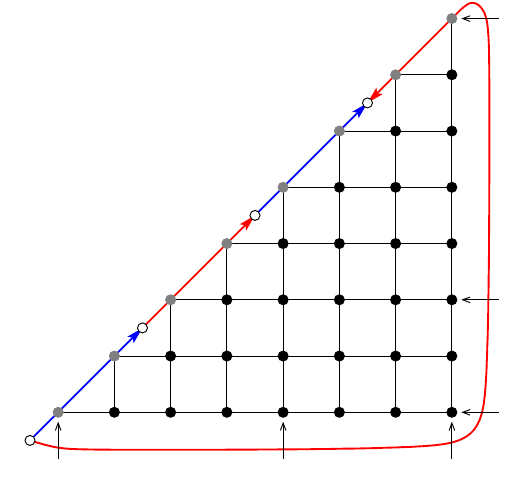}%
\end{picture}%
\setlength{\unitlength}{3947sp}%
\begingroup\makeatletter\ifx\SetFigFont\undefined%
\gdef\SetFigFont#1#2#3#4#5{%
  \reset@font\fontsize{#1}{#2pt}%
  \fontfamily{#3}\fontseries{#4}\fontshape{#5}%
  \selectfont}%
\fi\endgroup%
\begin{picture}(4080,3882)(136,-3046)
\put(2351,-2986){\makebox(0,0)[lb]{\smash{{\SetFigFont{9}{10.8}{\familydefault}{\mddefault}{\updefault}{\color[rgb]{0,0,0}$i$}%
}}}}
\put(3526,-2986){\makebox(0,0)[lb]{\smash{{\SetFigFont{9}{10.8}{\familydefault}{\mddefault}{\updefault}{\color[rgb]{0,0,0}$|\sigma|m-1$}%
}}}}
\put(4201,664){\makebox(0,0)[lb]{\smash{{\SetFigFont{9}{10.8}{\familydefault}{\mddefault}{\updefault}{\color[rgb]{0,0,0}$|\sigma|m-1$}%
}}}}
\put(2251,-1711){\makebox(0,0)[lb]{\smash{{\SetFigFont{9}{10.8}{\familydefault}{\mddefault}{\updefault}{\color[rgb]{0,0,0}$(i,j)$}%
}}}}
\put(1576,-1261){\makebox(0,0)[lb]{\smash{{\SetFigFont{11}{13.2}{\familydefault}{\mddefault}{\updefault}{\color[rgb]{1,0,0}$a_2$}%
}}}}
\put(676,-2161){\makebox(0,0)[lb]{\smash{{\SetFigFont{11}{13.2}{\familydefault}{\mddefault}{\updefault}{\color[rgb]{0,0,1}$a_1$}%
}}}}
\put(2476,-361){\makebox(0,0)[lb]{\smash{{\SetFigFont{11}{13.2}{\familydefault}{\mddefault}{\updefault}{\color[rgb]{0,0,1}$a_1$}%
}}}}
\put(3376,539){\makebox(0,0)[lb]{\smash{{\SetFigFont{11}{13.2}{\familydefault}{\mddefault}{\updefault}{\color[rgb]{1,0,0}$a_2$}%
}}}}
\put(151,-2611){\makebox(0,0)[lb]{\smash{{\SetFigFont{11}{13.2}{\familydefault}{\mddefault}{\updefault}{\color[rgb]{0,0,0}$c_0$}%
}}}}
\put(1051,-1711){\makebox(0,0)[lb]{\smash{{\SetFigFont{11}{13.2}{\familydefault}{\mddefault}{\updefault}{\color[rgb]{0,0,0}$c_1$}%
}}}}
\put(1951,-811){\makebox(0,0)[lb]{\smash{{\SetFigFont{11}{13.2}{\familydefault}{\mddefault}{\updefault}{\color[rgb]{0,0,0}$c_2$}%
}}}}
\put(2851, 89){\makebox(0,0)[lb]{\smash{{\SetFigFont{11}{13.2}{\familydefault}{\mddefault}{\updefault}{\color[rgb]{0,0,0}$c_3$}%
}}}}
\put(551,-2986){\makebox(0,0)[lb]{\smash{{\SetFigFont{9}{10.8}{\familydefault}{\mddefault}{\updefault}{\color[rgb]{0,0,0}$0$}%
}}}}
\put(4201,-1586){\makebox(0,0)[lb]{\smash{{\SetFigFont{9}{10.8}{\familydefault}{\mddefault}{\updefault}{\color[rgb]{0,0,0}$j$}%
}}}}
\put(4201,-2486){\makebox(0,0)[lb]{\smash{{\SetFigFont{9}{10.8}{\familydefault}{\mddefault}{\updefault}{\color[rgb]{0,0,0}$0$}%
}}}}
\end{picture}%
}
  \caption{The polygonal embedding $\PP$ for $m = 2$ and $\sigma = a_1 a_2 a_1 \olsi{a}_2$, a non-canonical signature for a non-orientable surface of Euler genus~$2$ (the Klein bottle). Its size is $(2,28)$. The cycle outerface is $c_0 - (0,0) - (1,1) - c_1 - (2,2) - (3,3) - c_2- (4,4) - (5,5) - c_3 - (6,6) - (7,7)$, where $c_i$'s are the four corners of $\PP$.}
  \label{fig:pu}
\end{myfigure}

Let $\Pi$ be a polygonal embedding of size $(m,0)$.  Up to taking a p-major of $\Pi$, we can suppose that each side of $\Pi$ is composed of exactly $m$ non-corner vertices.

The embeddings $\Pi$ and $\PP$ share the same signature and each have $m$ non-corner vertices per side. Consequently, both embeddings have isomorphic borders and a cycle outerface consisting of $\card{\sigma(\Pi)}(m+1)$ vertices.

We denote by $(u_0,\dots,u_{|\sigma(\Pi)|-1})$ the border of $\Pi$, and by $v_0,\dots,v_{|\sigma(\Pi)|m-1}$ its non-corners vertices ordered clockwise around its cycle outerface, so that this cycle is
 \[
  u_0 - v_0 - v_1 - \cdots - v_{m-1} - u_1 - v_m - \cdots - v_{im-1} - u_i - v_{im} - \cdots
\]
Up to taking a p-major of $\Pi$, we will assume that all inner faces of $\Pi$ are triangulated such that corners have degree two. 
For each non-corner $v_i$, we define the two indices $a(i),b(i)$ such that $\set{v_k : k\in [a(i),b(i)] }$ is the minimal subset of vertices containing $v_i$ and all its non-corner neighbors. 
Note that $a(i) = i$ or $b(i) = i$ is possible. For technical reasons, we increase $b(0)$ to $b(0) = |\sigma(\Pi)|m$. 

In order to show that $\Pi$ is a minor of $\PP$, we construct a \emph{witness} of $\Pi$ in $\PP$ with the property that it preserves its outerface and its border. 
A witness for a minor $\Pi$ in $\PP$ is a collection $\set{W(u)}_{u\in V(\Pi)}$ of nonempty pairwise disjoint subsets of vertices of $\PP$ such that each set $W(u)$ induces a connected subgraph of $\PP$, and for each edge $u-v$ of $\Pi$, $W(u) \cup W(v)$ induces a connected component in $\PP$.

The witness of $\Pi$ in $\PP$ is defined by:
\begin{itemize}[noitemsep]
  \item $W(u_i) = \set{c_i}$ for each corner $u_i$; and
  \item $W(v_i) = \set{(i,j)} \cup C(v_i) \cup R(v_i)$ for each non-corner $v_i$, \\
  where $C(v_i) = \set{ (i,k) : k\in (a(i),i] }$ and $R(v_i) = \set{ (k,i) : k \in [i,b(i)) }$. %
\end{itemize}

Clearly, for corner vertices, $W(u_i) = \set{c_i}$ satisfies all required properties: non-emptiness, pairwise disjointness, and connectedness. For non-corner vertices, $W(v_i)$ induces a connected component of $\PP$ since it is the union of a column subpath $C(v_i)$ and of a row subpath $R(v_i)$ of $\PP$ intersecting in $(i,i)$ (in the case where $C(v_i)$ and $R(v_i)$ are not empty).

Let us check that $W(v_i)$ and $W(v_j)$ are pairwise disjoint (witness for corner and non-corner are clearly disjoint). 
Assume $j<i$. A non-empty intersection is only possible between $C(v_i)$ and $R(v_j)$. Moreover, this can only occur at $(i,j)$. 
Note that the case $j = 0$ provides disjoint witnesses, since $(i,0) \notin C(v_i)$. If $(i,j) \in C(v_i)$, then $j\in(a(i),i]$, and thus $v_i$ has a neighbor $v_{a(i)}$ with $a(i) < j$. If $(i,j) \in R(v_j)$, then $i\in [j,b(j))$, and thus $v_j$ has a neighbor $v_{b(j)}$ with $i < b(j)$. 
It follows that $a(i) < j < i < a(j)$ and that $v_{a(i)} - v_i$ and $v_j - v_{b(j)}$ is a pair of crossing edges, which is impossible in the outerplanar embedding $\Pi$.

It remains to check that, for each $u-v\in E(\Pi)$, $W(u) \cup W(v)$ induces a connected graphs in $\PP$. 
We can restrict our attention to non-corner neighbors, since for a corner $u_i$, $c_i \in W(u_i)$ have exactly two neighbors that are $(im-1,im-1) \in W(v_{im-1})$ and $(im,im) \in W(v_{im})$ (indices modulo $|\sigma(\Pi)|m$).

So consider an edge $v_i-v_j$ of $\Pi$ with $j<i$. Let $\sW_{i,j}$ be the set of vertices of the path connecting $(i,i) \in W(v_i)$ to $(j,j) \in W(v_j)$ inside the grid part of $\PP$, and defined by
\[
  \sW_{i,j} ~=~ \set{ (i,i), (i,i-1), \dots, (i,j+1), ~(i,j),~  (i-1,j), \dots,(j+1,j), (j,j) } ~.
\]
To show that $W(v_i) \cup W(v_j)$ induces a connected component in $\PP$, it suffices to show that $\sW_{i,j}\subseteq C(v_i) \cup R(v_j)$ since $\sW_{i,j}$ induces a path in $\PP$, and $C(v_i) \subset W(v_i)$ and $R(v_j) \subset W(v_j)$. 
The adjacency of $v_i$ and $v_i$ implies that $b(j)\geq i$ and $a(i)\leq j$, thus all the vertices $(i',j')$ for some $i'\in[j,i)$ or $j'\in[j,i)$ belong to $R(v_j) \subset W(v_j)$. 

It remains to check that $(i,j) \in C(v_i) \cup R(v_j)$, even if $(i,j) \notin C(v_i)$ or $(i,j) \notin R(v_j)$ is possible. If $j=0$, then we are done because we have set $b(0) = |\sigma(\Pi)|m$ and thus $(i,0) \in R(v_0)$. 
We are also done if $|i-j| = 1$ since in this case $(i,i) \in W(v_i)$ and $(j,j)\in W(v_j)$ are neighbors in $\PP$. 
Since the inner faces of $\Pi$ are triangulated, the edge $v_i-v_j$ (that lies inside $\Pi$ since $v_i$ and $v_j$ are not consecutive in the cycle outerface) shares a triangle $v_k-v_j-v_i$ such that $k \notin[j,i]$, the case $j = 0$ being excluded. 
If $k<j$, then because of the edge $v_i-v_k$, $a(i)\le k < j < i$, that is $j \in (a(i),i]$, and thus $(i,j) \in C(v_i)$. If $k>i$, then because of the edge $v_j-v_k$, $b(j)\ge k > i > j$, that is $i\in [j,b(j))$, and thus $(i,j) \in R(v_j)$.

We have therefore proved that $\sW_{i,j}\subseteq W(v_i) \cup W(v_j)$ and that $W(v_i)\cup W(v_j)$ induces a connected component in $\PP$. It follows that $\Pi$ is a minor of $\PP$ with same outerface, border and signature.  By \cref{property:p-minorOuterface}, $\Pi$ is a p-minor of $\PP$ that completes the proof of \cref{lem:final}.


\section{Conclusion}

We have shown that the family of Euler genus-$g$ graphs have an Euler genus-$g$ minor-universal graph of $O(g^2(n+g)^2)$ vertices. This polynomial construction appears to be a key point in the recent polynomial bounds for the Graph Minor Structure Theorem due to Gorsky, Seweryn and Wiederrecht~\cite{GSW25a,GSW25}.

It is not clear whether the dependencies in $g$ and in $n$ are tight. Even in the planar case, it is not obvious that the $n^2$ term is required, as the known $\Omega(n^2)$ lower bound applies only under the restriction that the minor-universal graph is a grid.

It would be interesting to extend this result to other graph families, in particular bounded tree-width graphs, $k$-planar graphs, and minor-free graphs. Another direction could be to extend the result to topological minor-universal graphs.

\paragraph{Acknowledgement.}

We are grateful to Arnaud de Mesmay for prelimary discussions and helpful suggestions regarding this problem, and to Niloufar Fuladi, Sebastian Wiederrecht and Maximilian Gorsky for complementary discussions.


\bibliographystyle{my_alpha_doi}
\bibliography{biblio}

\end{document}